\documentclass[a4,oneside,english,reqno,12pt]{amsart}
\usepackage{extsizes}
\usepackage{geometry}
\usepackage{amssymb,amsmath,amsthm,amsfonts,mathtools,esint,dsfont,bbm,yhmath}
\usepackage{appendix}
\usepackage{enumitem}
\usepackage{datetime}
\DeclareMathOperator{\tr}{Tr}
\usepackage[pdfencoding=auto,psdextra]{hyperref}
\usepackage{bookmark}
\usepackage{cleveref}
\usepackage{crossreftools}
\crefname{subsection}{subsection}{subsections}
\usepackage[foot]{amsaddr}
\usepackage{orcidlink}
\makeatletter

\usepackage{hyperref} 
\hypersetup{
  pdftitle=Homogeneous attractive Bose--Einstein condensates  with repulsive three-body interactions: the two-dimensional case,
  pdfauthor=Dinh-Thi Nguyen,
  pdfsubject=2D Homogeneous BECs,
}

\theoremstyle{plain}
\newtheorem{theorem}{Theorem}[section]
\newtheorem{lemma}[theorem]{Lemma}
\numberwithin{equation}{section}

\theoremstyle{remark}
\newtheorem{remark}[theorem]{Remark}
\newtheorem*{remark*}{Remark}

\newcommand{\wto}{\rightharpoonup}
\allowdisplaybreaks

\title[2D homogeneous BECs]{Homogeneous attractive Bose--Einstein condensates \\ with repulsive three-body interactions:\\ the two-dimensional case}

\author[D.-T. Nguyen]{Dinh-Thi Nguyen\,\orcidlink{0000-0003-4487-5557}}
\address{Faculty of Mathematics and Computer Science, University of Science, Ho Chi Minh City, Vietnam; and Vietnam National University, Ho Chi Minh City, Vietnam.}
\email{\href{mailto:ndthi@hcmus.edu.vn}{ndthi@hcmus.edu.vn}}


\begin{document}
\begin{abstract}
We consider a homogeneous Bose gas composed of $N$ identical bosons occupying an infinite two-dimensional space. This gas exhibits both an attractive two-body interaction and a repulsive three-body interaction. Via the intermediate Hartree theory, we rigorously derive the homogeneous cubic-quintic nonlinear Schrödinger functional as the mean-field limit of the model. Our investigation focuses on the system’s behavior in relation to the two-body interaction.
\end{abstract}

\maketitle

\section{Introduction}

In this paper, we investigate the ground state properties of a two-dimensional homogeneous Bose gas with attractive two-body interactions and repulsive three-body interactions. Our aim is to establish a rigorous connection between the quantum many-body Hamiltonian and its effective mean-field descriptions, and to analyze the regimes in which Bose--Einstein condensations (BECs) and energy convergence can be observed. We here consider a system of $N\geq 3$ identical bosons in $\mathbb{R}^{2}$, described by the nonrelativistic Hamiltonian
\begin{align}\label{hamiltonian}
H_{a,b,N} := & \sum_{i=1}^{N} -\Delta_{x_i} - \frac{a}{N-1}\sum_{1\leq i<j\leq N} U^{\rm 2B}_{N^\alpha}(x_i-x_j) \nonumber \\
& + \frac{b}{(N-1)(N-2)} \sum_{1\leq i<j<k\leq N} U^{\rm 3B}_{N^\beta}(x_i-x_j) U^{\rm 3B}_{N^\beta}(x_i-x_k) ,
\end{align}
acting on the symmetric space $L^{2}_{\rm sym}(\mathbb{R}^{2N})$. The choices of coupling constants $1/(N-1)$ and $1/(N-1)(N-2)$ ensure that the kinetic and interaction energies are comparable in the mean-field limit $N \to \infty$. The parameters $a \geq 0$ and $b \geq 0$ are the strengths of the two-body and three-body interactions, respectively. The interaction terms $U^{\rm 2B}_{N^\alpha}$ and $U^{\rm 3B}_{N^\beta}$ are scaled through the parameters $\alpha,\beta>0$, i.e.,
\begin{equation}\label{scaled-interaction}
U^{\rm 2B}_{N^\alpha}(x) := N^{2\alpha}U^{\rm 2B}(N^\alpha x) \quad \text{and} \quad U^{\rm 3B}_{N^\beta}(x) : = N^{2\beta}U^{\rm 3B}(N^\beta x),
\end{equation}
in such a way that we may expect a well-defined semiclassical theory in the limit $N\to\infty$. Here the distance function $U^{\rm 2B}$ and $U^{\rm 3B}$ satisfy the following common conditions
\begin{equation}\label{condition:2-3-body}
0 \leq U^{k\rm B}(x)=U^{k\rm B}(-x) \quad \text{ with } \quad U^{k\rm B},\widehat{U^{k\rm B}} \in L^1(\mathbb{R}^{2}) \quad \text{ and } \quad \int_{\mathbb{R}^{2}}U^{k\rm B}(x) {\rm d}x = 1,
\end{equation}
where $k \in \{2,3\}$. In the literature \cite{NamRicTri-23,NguRic-24,NguRic-25}, the three-body repulsion was represented by a function, namely $0 \leq W \in L^{1}(\mathbb R^{4})$, of distance between three arbitrary variables with symmetry conditions, i.e.,
\begin{equation}\label{3-body-interaction}
W(x,y) = W(y,x) \quad \text{and} \quad W(x-y,x-z) = W(y-x,y-z) = W(z-x,z-y).
\end{equation}
However, due to technical issues arising in the mathematical studies of homogeneous Bose gases, we consider here the three-body interaction as the product of the same two-body interaction \eqref{condition:2-3-body} which preserves the properties of three-body interactions. This allows us to use a variant version of Onsager’s lemma in order to deal with the three-body interaction (see \Cref{lem:three-one} below) analogously to the two-body interaction. It is worth noting that such a consideration of three-body interactions is somehow similar to the circumradius of the Menger--Melnikov curvature in the study of anyon gases \cite{LunRou-15}.

In this article, our focus lies on the mean-field limit of the quantum ground state energy per particle, given by
\begin{equation}\label{energy:quantum}
E_{a,b,N}^{\mathrm{QM}} := \inf\left\{\left\langle\Psi_N \left|\frac{H_{a,b,N}}{N}\right| \Psi_N \right\rangle : \Psi_N \in L^{2}_{\rm sym}(\mathbb{R}^{2N}), \int_{\mathbb R^{2N}}|\Psi_N|^{2}=1 \right\}
\end{equation}
and on the associated ground states. In the many-body setting, the mean-field ansatz
\begin{equation}\label{eq:BEC}
\Psi_N(x_1,\ldots,x_N) \approx v^{\otimes N}(x_1,\ldots,x_N) := v(x_1)\ldots v(x_N) \quad \text{with} \quad \|v\|_{L^{2}}^{2}=1
\end{equation}
leads to the homogeneous Hartree functional
\begin{align}\label{functional:h}
\mathcal{E}^{\rm H}_{a, b,N}[v] :={} & \int_{\mathbb{R}^{2}}|\nabla v(x)|^{2} {\rm d}x - \frac{a}{2}\iint_{\mathbb{R}^{4}} U^{\rm 2B}_{N^\alpha}(x-y)|v(x)|^{2}|v(y)|^{2} {\rm d}x {\rm d}y \nonumber \\
& + \frac{b}{6}\iiint_{\mathbb{R}^6} U^{\rm 3B}_{N^\beta}(x-y)U^{\rm 3B}_{N^\beta}(x-z)|v(x)|^{2}|v(y)|^{2}|v(z)|^{2} {\rm d}x {\rm d}y {\rm d}z .
\end{align}
Its minimum provides an upper bound on the many-body ground state energy, i.e.,
\begin{equation}\label{energy:h}
E_{a,b,N}^{\mathrm{H}}[1] := \inf\left\{\mathcal{E}^{\rm H}_{a, b,N}[v]:v\in H^1(\mathbb{R}^{2}), \int_{\mathbb R^{2}}|v|^{2}=1\right\} \geq E_{a,b,N}^{\rm QM}.
\end{equation}
The Hartree functional \eqref{functional:h}, which is commonly used to describe the mean-field approximation in many-body systems, can be formally related to the NLS functional under certain conditions. In the large-$N$ limit, the rescaled potentials converge to delta distributions, and the Hartree functional reduces to the cubic-quintic nonlinear Schr\"odinger (NLS) functional
\begin{equation}\label{functional:nls}
\mathcal{E}^{\rm NLS}_{a,b}[v] := \int_{\mathbb{R}^{2}}\left[|\nabla v(x)|^{2} - \frac{a}{2} |v(x)|^{4} + \frac{b}{6} |v(x)|^6\right] {\rm d}x ,
\end{equation}
with the associated NLS ground state energy
\begin{equation}\label{energy:nls}
E^{\rm NLS}_{a,b}[1] := \inf\left\{\mathcal{E}^{\rm NLS}_{a,b}[v]:v\in H^1(\mathbb{R}^{2}), \int_{\mathbb{R}^{2}}|v(x)|^{2}=1\right\}.
\end{equation}

The main contribution of this work is to rigorously justify this hierarchy of models: from the many-body Hamiltonian \eqref{hamiltonian} to the Hartree functional \eqref{functional:h}, and further to the cubic-quintic NLS functional \eqref{functional:nls}. We establish asymptotics of the ground state energy in the large-$N$ limit and characterize the qualitative behavior of ground states across different parameter regimes. Our results extend previous works on the cubic NLS theory for inhomogeneous Bose gases \cite{Lewin-ICMP,LewNamRou-16,LewNamRou-17,NamRou-20,CheHol-17,LewNamRou-17-proc,GuoSei-14,Nguyen-20,GuoLuoYan-20,DinNguRou-24} (see also \cite{BraSacTolHul-95,BraSacHul-97,SacStoHul-98}), as well as studies of Bose systems with three-body interactions \cite{ChePav-11,Chen-12,Yuan-15,Xie-15,CheHol-19,NamSal-20,LiYao-21,NamRicTri-22a,NamRicTri-22b,NamRicTri-23,NguRic-24,NguRic-25} (see also \cite{EsrGreZhoLin-96,JosRic-97,GamFreTom-99,AkhDasVag-99,GamFreTomCho-00}). In particular, we demonstrate how the repulsive three-body term stabilizes the system and prevents collapse, leading to rich phenomena such as self-trapping \cite{Petrov-14}, supersolid phase \cite{BisBla-15,Blakie-16}, and phase transitions between cubic and cubic-quintic regimes.

In the absence of three-body interaction, it is well known that the NLS functional is stable if and only if the strength $a$ of the attractive two-body interaction is below the critical mass $a_{*}$, which is the optimal constant in the Gagliardo--Nirenberg inequality
\begin{equation}\label{ineq:gn}
\frac{a_{*}}{2} \int_{\mathbb R^{2}}|v(x)|^{4}{\rm d}x \leq \int_{\mathbb R^{2}}|v(x)|^{2}{\rm d}x \int_{\mathbb R^{2}}|\nabla v(x)|^{2}{\rm d}x.
\end{equation}
Moreover, it is well-known that \eqref{ineq:gn} has a positive radial symmetry optimizer $Q\in H^1(\mathbb{R}^{2})$, which is the unique optimizer up to translations, multiplication by a complex factor, and scaling. Such an optimizer solves the following cubic NLS equation in $\mathbb R^{2}$ (see \cite{Weinstein-83})
\begin{equation}\label{eq:nls}
-\Delta Q + Q - Q^{3} = 0.
\end{equation}
The critical value is then $a_{*} = \|Q\|_{L^{2}}^{2}$ where $Q$ is the unique (up to translations) positive radial solution of \eqref{eq:nls}. In order to formulate our results in this article, we define the $L^{2}(\mathbb R^{2})$-normalized unique solution of \eqref{eq:nls}, i.e.,
\begin{equation}\label{def:cubic-nls}
Q_0 := \|Q\|_{L^{2}}^{-1} Q.
\end{equation}
It then follows from \eqref{ineq:gn} and \eqref{eq:nls} that
\begin{equation}\label{norms_Q_0}
\int_{\mathbb R^{2}}|\nabla Q_{0}(x)|^{2}{\rm d}x = \frac{a_{*}}{2}\int_{\mathbb R^{2}}|Q_{0}(x)|^{4}{\rm d}x = \int_{\mathbb R^{2}}|Q_{0}(x)|^{2}{\rm d}x = 1.
\end{equation}

It is an interesting phenomenon that, in the limit of strong attractive two-body interaction, i.e., $a \to \infty$, the kinetic term in \eqref{functional:nls} is neglected, and the minimization problem \eqref{energy:nls} reduces to a local density approximation. This is the so-called ``Thomas--Fermi limit'' (TF) which appeared in the studies of confined BECs with repulsive two-body interactions \cite{LieSeiYng-00}. Indeed, it was proved in \cite{DoaNgu-26} that, in the Thomas--Fermi limit, we have
\begin{equation}\label{energy:TF}
E^{\rm NLS}_{a,b}[1] \approx E^{\rm TF}_{a,b}[1] = \inf\left\{\mathcal{E}^{\rm TF}_{a,b}[\varrho] : 0\leq \varrho \in L^{1} \cap L^{3}(\mathbb R^2), \int_{\mathbb R^2}\varrho(x){\rm d}x = 1\right\} = \frac{a^{2}}{b}E^{\rm TF}_{1,1}[1] < 0
\end{equation}
where
\begin{equation}\label{functional:tf}
\mathcal{E}^{\rm TF}_{a,b}[\varrho] = \int_{\mathbb R^2} \left[-\frac{a}{2}\varrho(x)^{2} + \frac{b}{6}\varrho(x)^{3}\right]{\rm d}x.
\end{equation}
Furthermore, the minimization problem $E^{\rm TF}_{a,b}[1]$ in \eqref{energy:TF} admits (up to translation) a unique ground state $\varrho^{\rm TF}_{a,b}$, which is positive radially symmetric decreasing and satisfies, by a simple scaling,
\begin{equation}\label{energy:TF-scaling}
\varrho^{\rm TF}_{a,b}(x) = \frac{a}{b}\varrho^{\rm TF}_{1,1}\left(\left(\frac{a}{b}\right)^{\frac{1}{2}}x\right).
\end{equation}
In the following, we summarize the results of \cite{DoaNgu-26} on the classification of existence and non-existence of homogeneous cubic-quintic NLS ground states as well as its asymptotic behaviors.

\begin{theorem}[Existence of NLS ground states]\label{thm:existence-nls}
Let $a>0$, $b\in\mathbb{R}$, and $E^{\rm NLS}_{a,b}[1]$ be given in \eqref{energy:nls}.
\begin{enumerate}[label=(\roman*)]
\item\label{nls-gs-inexistence} {\bf Non-existence of NLS ground states.} If either $b<0$ or $b=0$ and $a>a_{*}$ then $E^{\rm NLS}_{a,b}[1] = -\infty$. Furthermore, if $b \geq 0$ and $a \leq a_{*}$ then $E_{a, b}^{\rm NLS} = 0$ but there are no ground states, except the unique ground state $Q_{0}$ in \eqref{def:cubic-nls} of $E^{\rm NLS}_{a_{*},0}[1] = 0$.

\item\label{nls-gs-existence} {\bf Existence of NLS ground states.} If $b>0$ and $a > a_{*}$ then $E^{\rm NLS}_{a,b}[1] < 0$ and there are ground states, which are nonnegative radially symmetric decreasing.

\item {\bf Existence and uniqueness of TF ground state.} The minimization problem $E^{\rm TF}_{a,b}[1]$ given by \eqref{energy:TF} admits a (unique) ground state for every fixed $a,b>0$. Furthermore, $E^{\rm TF}_{a,b}[1] = -\dfrac{3a^{2}}{8b}$.
\end{enumerate}
\end{theorem}

\begin{theorem}[Mass concentration of NLS ground states]\label{thm:behavior-nls}
Let $\{a_{n}\},\{b_{n}\} \subset (0, \infty)$ be such that $a_{n} \to a_{0} \in [a_{*},\infty]$ and $b_{n} \to b_{0} \in [0,\infty]$. Let $\{v_{n}\}$ be a sequence of ground states of $E^{\rm NLS}_{a_{n},b_{n}}[1]$ given by \eqref{energy:nls}. We have the following.

\begin{enumerate}[label=(\roman*)]
\item {\bf Cubic-quintic approximation.} If $a_{*} < a_{0} < \infty$ then, up to a translation and extracting a subsequence,
\begin{equation}\label{cv:gs-cubic-quintic-2d}
\lim_{n\to\infty} b_{n}^{\frac{1}{2}}v_{n}(b_{n}^{\frac{1}{2}}\cdot) = w_{0}
\end{equation}
strongly in $H^{1}(\mathbb R^{2})$, where $w_{0}$ is a ground state of $E^{\rm NLS}_{a_{0},1}[1]$. Furthermore,
\begin{equation}\label{cv:energy-cubic-quintic-2d}
\lim_{n\to\infty} b_{n}E^{\rm NLS}_{a_{n},b_{n}}[1] = E^{\rm NLS}_{a_{0},1}[1].
\end{equation}

\item {\bf Cubic approximation.} If $a_{0} = a_{*}$ then, up to a translation,
\begin{equation}\label{cv:gs-cubic-2d}
\lim_{n\to\infty} \ell_{n}^{\frac{1}{2}}v_{n}(\ell_{n}^{\frac{1}{2}} \cdot) = Q_0 \quad \text{with} \quad \ell_{n} = \frac{2\|Q_{0}\|_{L^{6}}^{6}b_{n}}{3\|Q_{0}\|_{L^{4}}^{4}(a_{n} - a_{*})}
\end{equation}
strongly in $H^1(\mathbb{R}^{2})$ for the whole sequence. Here $Q_{0}$ is the unique $L^{2}$-normalized optimizer for \eqref{ineq:gn}. Furthermore,
\begin{equation}\label{cv:energy-cubic-2d}
\lim_{n\to\infty}\frac{b_{n}}{(a_{n}-a_{*})^{2}}E^{\rm NLS}_{a_{n},b_{n}}[1] = -\frac{3}{8}\frac{\|Q_{0}\|_{L^{4}}^{8}}{\|Q_{0}\|_{L^{6}}^{6}}.
\end{equation}

\item {\bf Thomas--Fermi approximation.} If $a_{0} = \infty$ then, up to a translation,
\begin{equation}\label{cv:gs-h-tf}
\lim_{n\to\infty} \frac{b_{n}}{a_{n}}v_{n}\left(\left(\frac{b_{n}}{a_{n}}\right)^{\frac{1}{2}} \cdot\right)^{2} = \varrho^{\rm TF}_{1,1}
\end{equation}
strongly in $L^{1}\cap L^{3}(\mathbb{R}^{2})$ for the whole sequence. Furthermore,
\begin{equation}\label{cv:energy-h-tf}
\lim_{n\to\infty}\frac{b_{n}}{a_{n}^{2}}E^{\rm NLS}_{a_{n},b_{n}}[1] = E^{\rm TF}_{1,1}[1].
\end{equation}
\end{enumerate}
\end{theorem}

In this work, we endeavor to rigorously derive the homogeneous 2D cubic-quintic nonlinear Schrödinger (NLS) theory as the mean-field model of Bose gases. As observed in the effective cubic-quintic NLS theory, the system exhibits perpetual stability due to the repulsive three-body interaction. Notably, without additional external potentials, a system characterized by an attractive two-body interaction and a repulsive three-body interaction exhibits self-trapping behavior. In such scenarios, the existence of NLS ground states is contingent upon the two-body interaction being sufficiently negative, as elucidated in \Cref{thm:existence-nls}. Consequently, the system's behaviors is heavily influenced by the attractive two-body interaction, as demonstrated in \Cref{thm:behavior-nls}.

In light of \Cref{thm:existence-nls,thm:behavior-nls}, we anticipate analogous outcomes in the Hartree theory \eqref{energy:h}. However, the existence of Hartree ground states is not immediately apparent due to the absence of its desirable property of radial symmetry and decreasing potential. This deficiency can be rectified through the concentration compactness method in the calculus of variations \cite{Lions-84a}. Furthermore, additional assumptions regarding the two-body and three-body functions are necessary to establish mass concentrations. Our findings encompass the following results.

\begin{theorem}[Existence of Hartree ground states]\label{thm:existence-h}
Let $0<\alpha<\beta$ and assume that $U^{\rm 2B}_{N^{\alpha}}$, $U^{\rm 3B}_{N^{\beta}}$ satisfy \eqref{scaled-interaction}-\eqref{condition:2-3-body}. Let $a>a_{*}$ and $b>0$ be fixed. Then, the minimization problems \eqref{energy:h} admit ground states for $N$ sufficient large.
\end{theorem}

\begin{theorem}[Mass concentration of Hartree ground states]\label{thm:behaviors-h}
Let $0<\alpha<\beta$ and assume that $U^{\rm 2B}_{N^{\alpha}}$, $U^{\rm 3B}_{N^{\beta}}$ satisfy \eqref{scaled-interaction}-\eqref{condition:2-3-body}. Let $\{a_{N}\} \subset (a_{*},\infty)$, $\{b_{N}\} \subset (0, \infty)$ be such that $a_{N} \to a_{0} \in [a_{*},\infty]$ and $b_{N} \to b_{0} \in [0,\infty]$, as $N\to\infty$. Let $\{v_{N}\}$ be a sequence of ground states of $E^{\rm H}_{a_{N},b_{N},N}[1]$ given by \eqref{energy:h}, for $N$ sufficient large. We have the following.
\begin{enumerate}[label=(\roman*)]
\item\label{thm:behavior-h-cubic-quintic} {\bf Cubic-quintic approximation.} Assume that $a_{*} < a_{0} < \infty$ and that $b_{N}^{-\frac{1}{2}}N^{-\alpha} \ll 1$. Then, up to a translation and extracting a subsequence,
\begin{equation}\label{cv:gs-h-cubic-quintic}
\lim_{N\to\infty} b_{N}^{\frac{1}{2}}v_{N}(b_{N}^{\frac{1}{2}}\cdot) = w_{0}
\end{equation}
strongly in $H^{1}(\mathbb R^{2})$, where $w_{0}$ is a ground state of $E^{\rm NLS}_{a_{0},1}[1]$. Furthermore,
\begin{equation}\label{cv:energy-h-cubic-quintic}
\lim_{N\to\infty} b_{N}E^{\rm H}_{a_{N},b_{N},N}[1] = E^{\rm NLS}_{a_{0},1}[1].
\end{equation}

\item\label{thm:behavior-h-cubic} {\bf Cubic approximation.} Assume that $a_{0} = a_{*}$ and that $N^{-\alpha}(a_{N}-a_{*})^{-\frac{1}{2}}b_{N}^{-\frac{1}{2}} \ll 1$ and $xU^{\rm 2B} \in L^{1}(\mathbb R^{2})$. Then, up to a translation,
\begin{equation}\label{cv:gs-h-cubic}
\lim_{N\to\infty} \ell_{N}^{\frac{1}{2}} v_{N}(\ell_{N}^{\frac{1}{2}} \cdot) = Q_0 \quad \text{with} \quad \ell_{N} = \frac{2\|Q_{0}\|_{L^{6}}^{6}b_{N}}{3\|Q_{0}\|_{L^{4}}^{4}(a_{N} - a_{*})}
\end{equation}
strongly in $H^1(\mathbb{R}^{2})$ for the whole sequence. Here $Q_{0}$ is the unique $L^{2}$-normalized optimizer for \eqref{ineq:gn}. Furthermore,
\begin{equation}\label{cv:energy-h-cubic}
\lim_{N\to\infty}\frac{b_{N}}{(a_{N}-a_{*})^{2}}E^{\rm H}_{a_{N},b_{N},N}[1] = -\frac{3\|Q_{0}\|_{L^{4}}^{8}}{8\|Q_{0}\|_{L^{6}}^{6}}.
\end{equation}

\item\label{thm:behavior-h-tf} {\bf Thomas--Fermi approximation.} Assume that $a_{0} = \infty$ and that $\dfrac{a_{N}^{2}}{b_{N}} \ll N^{2\alpha}$ and $\dfrac{b_{N}^{2}}{a_{N}^{2}}N^{5\alpha-\beta} \ll 1$. Then, up to a translation,
\begin{equation}\label{cv:gs-h-tf}
\lim_{N\to\infty} \frac{b_{N}}{a_{N}} v_{N}^{*}\left(\left(\frac{b_{N}}{a_{N}}\right)^{\frac{1}{2}} \cdot\right)^{2} = \varrho^{\rm TF}_{1,1}
\end{equation}
strongly in $L^{1}\cap L^{3}(\mathbb{R}^{2})$, where $v_{N}^{*}$ is the symmetric decreasing rearrangement of $v_{N}$. Furthermore,
\begin{equation}\label{cv:energy-h-tf}
\lim_{N\to\infty}\frac{b_{N}}{a_{N}^{2}}E^{\rm H}_{a_{N},b_{N},N}[1] = E^{\rm TF}_{1,1}[1].
\end{equation}
\end{enumerate}
\end{theorem}

\begin{remark}
\begin{itemize}
\item At fixed $a > a_{*}$ and $b > 0$, the convergence of $E^{\rm H}_{a,b,N}[1]$ to $E^{\rm NLS}_{a,b}[1]$ as well as of ground states, as $N \to \infty$, was covered by \eqref{cv:energy-h-cubic-quintic} and \eqref{cv:gs-h-cubic-quintic}. The condition $b^{-\frac{1}{2}}N^{-\alpha} \ll 1$ is then trivial.
\item In the TF limit, the convergence of Hartree ground states reduces to the convergence of its density as an ``approximate'' TF ground state. The latter is challenging to obtain due to the locality of the TF functional. Utilizing the desirable property of the radially symmetric decreasing function, we derive the convergence \eqref{cv:gs-h-tf} of the symmetric decreasing rearrangement of Hartree ground states instead of its original. It remains an open problem to obtain the convergence of original Hartree ground states.
\end{itemize}
\end{remark}

Following the observation of the physical phenomenon in the intermediate Hartree theory, we extend the results to the full many-body theory \eqref{hamiltonian}. In the studies of inhomogeneous Bose gases, the mathematical method employed is the quantum de Finetti theorem \cite{Rougerie-EMS}, as used in \cite{LewNamRou-16,LewNamRou-17,NamRou-20,Rougerie-20,NguRic-24,NguRic-25}. This approach involves localizing existing quantum ground states into finite-dimensional spaces using orthogonal projection of the one-body operator. However, this process necessitates a “nice” external potential. Consequently, the method of quantum de Finetti cannot be applied to the study of homogeneous BECs. For such homogeneous problems, the method of Onsager lemma is utilized, as exemplified in the rotating BECs study presented in \cite{LewNamRou-17-proc,DinNguRou-24}. While this method is effective, it incurs a trade-off of a slower convergence speed.

We shall demonstrate that the convergence of the quantum energy and the Hartree energy in the low-density regimes, specifically for $0<\alpha<\frac{1}{2}$ and $0<\beta<\frac{1}{4}$, is achieved. The condensation of the energy for higher values of $\alpha$ and $\beta$ requires advanced techniques. When comparing the homogeneous Hartree energy with the effective homogeneous cubic-quintic NLS energy, it is necessary that $\alpha<\beta$, as we have seen in \Cref{thm:existence-h,thm:behaviors-h}. This condition arises because the attractive two-body interaction is compensated by the three-body interaction in homogeneous Bose gases. This compensation is due to the fact that, for homogeneous Bose gases, the attractive two-body interaction must be sufficiently negative, i.e., $a>a_{*}$, and cannot be controlled by the kinetic term. 

It is noteworthy that a genuine many-body ground state for \eqref{hamiltonian} does not exist, primarily due to the translation invariance and the linearity of the many-body system. This phenomenon is encountered in various other models, such as \cite{DinNguRou-24,LieYau-87}. Although ``approximate'' ground states could be considered, their condensation is not anticipated due to their superposition within the many-body theory. A possibility proposed in \cite{LieYau-87} involves confining the translation-invariant system within the bounded domain of the limiting profile. Notably, this scenario aligns with our TF theory \eqref{functional:tf}. Subsequently, a Feynman--Hellman argument could be employed to derive the mass concentration of “localized” approximate many-body ground states. However, this necessitates the convergence of ``approximate'' TF ground states, which is not immediately apparent due to the locality of the TF functional. Despite the absence of a condensate in the translation-invariant setting, the many-body ground state energy still converges to the corresponding effective mean-field model. Consequently, we have the following observations.

\begin{theorem}\label{thm:qm}
Under the same assumptions as in \Cref{thm:behaviors-h} and the additional assumptions that $\widehat{U^{\rm 3B}} \geq 0$ and $a_{N}N^{2\alpha-1} + b_{N}N^{4\beta-1} \ll \ell_{N}$, where
$$
\ell_{N} = 
\begin{cases}
\dfrac{1}{b_{N}} & \text{if } a_{*}<a_{0}<\infty, \\
\dfrac{(a_{N}-a_{*})^{2}}{b_{N}} & \text{if } a_{0} = a_{*}, \\
\dfrac{a_{N}^{2}}{b_{N}} & \text{if } a_{0} = \infty,
\end{cases}
$$
we have the convergence of the quantum ground state energy
$$
\lim_{N\to\infty}\ell_{N}^{-1}E_{a_{N},b_{N},N}^{\mathrm{QM}} = 
\begin{cases}
E^{\rm NLS}_{a_{0},1}[1] & \text{if } a_{*}<a_{0}<\infty, \\
-\dfrac{3}{8}\dfrac{\|Q_{0}\|_{L^{4}}^{8}}{\|Q_{0}\|_{L^{6}}^{6}} & \text{if } a_{0} = a_{*}, \\
E^{\rm TF}_{1,1}[1] & \text{if } a_{0} = \infty.
\end{cases}
$$
\end{theorem}
\medskip

The remainder of the paper presents the detailed mathematical proofs of our findings. In Appendix \ref{app:inequality}, we outline a functional inequality that proves instrumental in estimating the interaction potentials.

\medskip
\noindent{\bf Acknowledgement.} This research was funded by Vietnam Ministry of Education and Training under grant number B2026-CTT-10.

\section{A Hartree theory of homogeneous BECs with three-body interactions}

In this section, we consider the Hartree theory of homogeneous BECs. In particular, we consider \eqref{functional:h} and \eqref{energy:h}. Our concerns are the existence of Hartree ground states and its asymptotic behaviors in the large or small scaling limits. We will need preliminary estimates concerning the convergence (rate) of two-body and three-body interactions. While the following estimates \eqref{cv:h-nls-3-body}, \eqref{cv-rate:h-nls-3-body} are useful for the original three-body interaction in \eqref{functional:h}, we will need the similar estimates \eqref{cv:h-nls-3-body-modified}, \eqref{cv-rate:h-nls-3-body-modified} in the modified system \eqref{energy:h-modified} in the next section.

\begin{lemma}\label{lem:cv-hartree-potentials}
Let $\alpha,\beta>0$ and assume that $U^{\rm 2B}_{N^{\alpha}}$, $U^{\rm 3B}_{N^{\beta}}$ satisfy \eqref{scaled-interaction}-\eqref{condition:2-3-body}. Then, for every $v\in H^{1}(\mathbb R^{2})$,
\begin{align}
\int_{\mathbb R^{2}}|v(x)|^4{\rm d}x & = \lim_{N\to\infty}\iint_{\mathbb R^{4}} U^{\rm 2B}_{N^{\alpha}}(x-y)|v(x)|^{2}|v(y)|^{2}{\rm d}x{\rm d}y, \label{cv:h-nls-2-body} \\
\int_{\mathbb R^{2}}|v(x)|^{6}{\rm d}x & = \lim_{N\to\infty}\iiint_{\mathbb R^{6}} U^{\rm 3B}_{N^{\beta}}(x-y)U^{\rm 3B}_{N^{\beta}}(x-z)|v(x)|^{2}|v(y)|^{2}|v(z)|^{2}{\rm d}x{\rm d}y{\rm d}z \label{cv:h-nls-3-body} \\
& = \lim_{N\to\infty}\iint_{\mathbb R^{4}}|v(x)|^{2}\sqrt{U^{\rm 3B}_{N^{\beta}}*|v|^{2}}(x)U^{\rm 3B}_{N^{\beta}}(x-y)|v(y)|^{2}\sqrt{U^{\rm 3B}_{N^{\beta}}*|v|^{2}}(y){\rm d}x{\rm d}y. \label{cv:h-nls-3-body-modified}
\end{align}
Assume further that $xU^{\rm 2B}(x),xU^{\rm 3B}(x) \in L^{1}(\mathbb R^{2})$ then
\begin{align}
0 \leq{}& \int_{\mathbb R^{2}}|v(x)|^4{\rm d}x - \iint_{\mathbb R^{4}} U^{\rm 2B}_{N^{\alpha}}(x-y)|v(x)|^{2}|v(y)|^{2}{\rm d}x{\rm d}y \label{cv-rate:h-nls-2-body-0} \\
\leq{}& 2N^{-\alpha} \|v\|_{L^{6}}^{3} \|\nabla v\|_{L^{2}} \int_{\mathbb R^{2}} |xU^{\rm 2B}(x)| {\rm d}x, \label{cv-rate:h-nls-2-body} \\
0 \leq{}& \int_{\mathbb R^{2}}|v(x)|^{6} {\rm d}x - \iiint_{\mathbb R^{6}} U^{\rm 3B}_{N^{\beta}}(x-y)U^{\rm 3B}_{N^{\beta}}(x-z) |v(x)|^{2} |v(y)|^{2} |v(z)|^{2} {\rm d}x {\rm d}y {\rm d}z \label{cv-rate:h-nls-3-body-0} \\
\leq{}& 4N^{-\beta} \|v\|_{L^{10}}^{5} \|\nabla v\|_{L^{2}} \int_{\mathbb R^{2}} |xU^{\rm 3B}(x)| {\rm d}x, \label{cv-rate:h-nls-3-body} \\
0 \leq{}& \int_{\mathbb R^{2}}|v(x)|^{6} {\rm d}x - \iint_{\mathbb R^{4}}|v(x)|^{2}\sqrt{U^{\rm 3B}_{N^{\beta}}*|v|^{2}}(x)U^{\rm 3B}_{N^{\beta}}(x-y)|v(y)|^{2}\sqrt{U^{\rm 3B}_{N^{\beta}}*|v|^{2}}(y){\rm d}x{\rm d}y \label{cv-rate:h-nls-3-body-modified-0} \\
\leq{}& 5N^{-\beta} \|v\|_{L^{10}}^{5} \|\nabla v\|_{L^{2}} \int_{\mathbb R^{2}} |xU^{\rm 3B}(x)| {\rm d}x. \label{cv-rate:h-nls-3-body-modified}
\end{align}
\end{lemma}
\begin{proof}
The identity \eqref{cv:h-nls-2-body} can be found in \cite[Lemma 7]{LewNamRou-17}. We prove \eqref{cv:h-nls-3-body} by the same arguments. We first use H\"older and Young inequalities to obtain
\begin{align}\label{cv:h-nls-3-body-positive}
\iiint_{\mathbb R^{6}} U^{\rm 3B}_{N^{\beta}}(x-y)U^{\rm 3B}_{N^{\beta}}(x-z)|v(x)|^{2}|v(y)|^{2}|v(z)|^{2}{\rm d}x{\rm d}y{\rm d}z & = \int_{\mathbb R^{2}} |v(x)|^{2}(U^{\rm 3B}_{N^{\beta}}*|v|^{2})(x)^{2}{\rm d}x \nonumber \\
& \leq \|v\|_{L^{6}}^{2}\|U^{\rm 3B}_{N^{\beta}}*|v|^{2}\|_{L^{3}}^{2} \nonumber \\
& \leq \|v\|_{L^{6}}^{6}.
\end{align}
Here we have used the fact that 
\begin{equation}\label{3body:normalized}
\int_{\mathbb{R}^{2}} U^{\rm 3B}_{N^{\beta}}(x) {\rm d}x = \int_{\mathbb{R}^{2}} U^{\rm 3B}(x) {\rm d}x =1.
\end{equation}
In order to prove the reverse inequality, we use the changes of variables $y \mapsto x - N^{-\beta}y$ and rewrite
\begin{align}\label{cv:h-nls-3-body-1}
& \int_{\mathbb R^{2}}|v(x)|^{6}{\rm d}x - \iiint_{\mathbb R^{6}} U^{\rm 3B}_{N^{\beta}}(x-y)U^{\rm 3B}_{N^{\beta}}(x-z) |v(x)|^{2} |v(y)|^{2} |v(z)|^{2} {\rm d}x {\rm d}y {\rm d}z \nonumber \\
& = \int_{\mathbb R^{2}} \left[|v(x)|^{6} - |v(x)|^{2}(U^{\rm 3B}_{N^{\beta}}*|v|^{2})(x)^{2}\right]{\rm d}x \nonumber \\
& = \int_{\mathbb R^{2}} |v(x)|^{2} \left[|v(x)|^{2} + (U^{\rm 3B}_{N^{\beta}}*|v|^{2})(x)\right] \left[|v(x)|^{2} - (U^{\rm 3B}_{N^{\beta}}*|v|^{2})(x)\right] {\rm d}x \nonumber \\
& = \iint_{\mathbb R^{4}} |v(x)|^{2} \left[|v(x)|^{2} + (U^{\rm 3B}_{N^{\beta}}*|v|^{2})(x)\right] \left[|v(x)|^{2} - |v(x-N^{-\beta}y)|^{2}\right] U^{\rm 3B}(y) {\rm d}x{\rm d}y.
\end{align}
As in \cite{LewNamRou-17}, we split the integral over $y$ into two parts. We pick $L>0$ and decompose
\begin{equation}\label{decomposition:three-body}
U^{\rm 3B}(y) = \mathbbm{1}(|y|>L)U^{\rm 3B}(y) + \mathbbm{1}(|y|\leq L)U^{\rm 3B}(y).
\end{equation}
On the one hand, we use H\"older, Minkowski and Young inequalities to obtain
\begin{align}\label{cv:h-nls-3-body-2}
& \iint_{\mathbb R^{4}} |v(x)|^{2} \left[|v(x)|^{2} + (U^{\rm 3B}_{N^{\beta}}*|v|^{2})(x)\right] \left[|v(x)|^{2} - |v(x-N^{-\beta}y)|^{2}\right] \mathbbm{1}(|y|>L)U^{\rm 3B}(y) {\rm d}x{\rm d}y \nonumber \\
& \leq \int_{\mathbb R^{2}}\|v\|_{L^{6}}^{2} \||v|^{2}+U^{\rm 3B}_{N^{\beta}}*|v|^{2}\|_{L^{3}} \||v|^{2}-|v(\cdot-N^{-\beta}y)|^{2}\|_{L^{3}}\mathbbm{1}(|y|>L)U^{\rm 3B}(y) {\rm d}y \nonumber \\
& \leq 4\|v\|_{L^{6}}^{6}\int_{\mathbb R^{2}}\mathbbm{1}(|y|>L)U^{\rm 3B}(y) {\rm d}y,
\end{align}
where we have used again the normalization \eqref{3body:normalized}. On the other hand, we use
the diamagnetic inequality to obtain
\begin{align*}
\left||v(x)|^{2} - |v(x-N^{-\beta}y)|^{2}\right| &= \left|\int_{0}^{1} N^{-\beta}y \cdot \nabla |v|^{2}(x-tN^{-\beta}y) {\rm d}t\right| \\
& \leq 2N^{-\beta}|y| \int_{0}^{1} |v(x-tN^{-\beta}y)||\nabla v(x-tN^{-\beta}y)| {\rm d}t.
\end{align*}
This yields that
\begin{align}\label{cv:h-nls-3-body-3}
& \iint_{\mathbb R^{4}} |v(x)|^{2} \left[|v(x)|^{2} + (U^{\rm 3B}_{N^{\beta}}*|v|^{2})(x)\right] \left[|v(x)|^{2} - |v(x-N^{-\beta}y)|^{2}\right] \mathbbm{1}(|y|\leq L)U^{\rm 3B}(y) {\rm d}y{\rm d}x \nonumber \\
& \leq 2N^{-\beta}L \int_{0}^{1}\iint_{\mathbb R^{4}} |v(x)|^{2} \left[|v(x)|^{2} + (U^{\rm 3B}_{N^{\beta}}*|v|^{2})(x)\right] \times \nonumber \\
& \quad \times |v(x-tN^{-\beta}y)||\nabla v(x-tN^{-\beta}y)| U^{\rm 3B}(y) {\rm d}x{\rm d}y{\rm d}t \nonumber \\
& \leq 2N^{-\beta}L \int_{0}^{1}\int_{\mathbb R^{2}}\|v\|_{L^{10}}^{2} \||v|^{2}+U^{\rm 3B}_{N^{\beta}}*|v|^{2}\|_{L^{5}} \|v(\cdot-tN^{-\beta}y)\|_{L^{10}} \|\nabla v(\cdot-tN^{-\beta}y)\|_{L^{2}} U^{\rm 3B}(y) {\rm d}y{\rm d}t \nonumber \\
& \leq 4N^{-\beta}L \|v\|_{L^{10}}^{5}\|\nabla v\|_{L^{2}},
\end{align}
where we have used again H\"older, Minkowski and Young inequalities as well as \eqref{3body:normalized}. Putting all together \eqref{cv:h-nls-3-body-1}, \eqref{cv:h-nls-3-body-2}, \eqref{cv:h-nls-3-body-3}, and choosing $1 \ll L = L_{N} \ll N^{\beta}$ (for example $L= N^{\frac{\beta}{2}}$) we obtain the desired convergence \eqref{cv:h-nls-3-body}. 

The convergence \eqref{cv:h-nls-3-body-modified} can be proved analogously. We first use Cauchy--Schwarz inequality and \eqref{cv:h-nls-3-body-positive} to obtain
\begin{align}\label{cv:h-nls-3-body-positive-modified}
& \iint_{\mathbb R^{4}}|v(x)|^{2}\sqrt{U^{\rm 3B}_{N^{\beta}}*|v|^{2}}(x)U^{\rm 3B}_{N^{\beta}}(x-y)|v(y)|^{2}\sqrt{U^{\rm 3B}_{N^{\beta}}*|v|^{2}}(y){\rm d}x{\rm d}y \nonumber \\
& \leq \iint_{\mathbb R^{4}}U^{\rm 3B}_{N^{\beta}}(x-y)|v(x)|^{2}\frac{(U^{\rm 3B}_{N^{\beta}}*|v|^{2})(x)+(U^{\rm 3B}_{N^{\beta}}*|v|^{2})(y)}{2}|v(y)|^{2}{\rm d}x{\rm d}y \nonumber \\
& = \int_{\mathbb R^{2}} |v(x)|^{2}(U^{\rm 3B}_{N^{\beta}}*|v|^{2})(x)^{2}{\rm d}x \\
& \leq \|v\|_{L^{6}}^{6}. \nonumber 
\end{align}
To prove the reverse inequality, we add and subtract $\int_{\mathbb R^{2}}|v(x)|^{4}(U^{\rm 3B}_{N^{\beta}}*|v|^{2})(x) {\rm d}x$ and process as follows. On the one hand, we use the decomposition \eqref{decomposition:three-body} to estimate
\begin{align*}
& \left| \int_{\mathbb R^{2}}|v(x)|^{6}{\rm d}x - \int_{\mathbb R^{2}}|v(x)|^{4}(U^{\rm 3B}_{N^{\beta}}*|v|^{2})(x) {\rm d}x \right| \\
& \leq \iint_{\mathbb R^{4}}|v(x)|^{4} \left| |v(x)|^{2} - |v(x-N^{-\beta}y)|^{2} \right| \mathbbm{1}(|y|>L)U^{\rm 3B}(y) {\rm d}x{\rm d}y \\
& \quad + 2N^{-\beta}L\int_{0}^{1}\iint_{\mathbb R^{4}} |v(x)|^{4} |v(x-tN^{-\beta}y)| |\nabla v(x-tN^{-\beta}y)| U^{\rm 3B}(y) {\rm d}x{\rm d}y{\rm d}t \\
& \leq \int_{\mathbb R^{2}}\|v\|_{L^{6}}^{4}\||v|^{2} - |v(\cdot - N^{-\beta}y)|^{2}\|_{L^{3}}\mathbbm{1}(|y|>L)U^{\rm 3B}(y){\rm d}y \\
& \quad + 2N^{-\beta}L \int_{0}^{1}\int_{\mathbb R^{2}}\|v\|_{L^{10}}^{4}\|v(\cdot-tN^{-\beta}y)\|_{L^{10}} \|\nabla v(\cdot-tN^{-\beta}y)\|_{L^{2}} U^{\rm 3B}(y) {\rm d}y{\rm d}t \\
& \leq 2\|v\|_{L^{6}}^{6}\int_{\mathbb R^{2}}\mathbbm{1}(|y|>L)U^{\rm 3B}(y){\rm d}y + 2N^{-\beta}L\|v\|_{L^{10}}^{5}\|\nabla v\|_{L^{2}}
\end{align*}
which converges to $0$, by choosing again $1 \ll L = L_{N} \ll N^{\beta}$. On the other hand, we use \cite[Lemma 7]{LewNamRou-17} to estimate
\begin{align*}
& \left| \int_{\mathbb R^{2}}|v(x)|^{4}(U^{\rm 3B}_{N^{\beta}}*|v|^{2})(x) {\rm d}x - \iint_{\mathbb R^{4}}|v(x)|^{2}\sqrt{U^{\rm 3B}_{N^{\beta}}*|v|^{2}}(x)U^{\rm 3B}_{N^{\beta}}(x-y)|v(y)|^{2}\sqrt{U^{\rm 3B}_{N^{\beta}}*|v|^{2}}(y){\rm d}x{\rm d}y \right| \\
& \leq \iint_{\mathbb R^{4}} \mathbbm{1}(|y|>L)U^{\rm 3B}(y) |v(x)|^{2}\sqrt{U^{\rm 3B}_{N^{\beta}}*|v|^{2}}(x) \times \\
& \quad \times \Big| |v(x)|^{2}\sqrt{U^{\rm 3B}_{N^{\beta}}*|v|^{2}}(x) - |v(x)|^{2}\sqrt{U^{\rm 3B}_{N^{\beta}}*|v|^{2}}(x-N^{-\beta}y) \Big| {\rm d}x{\rm d}y \\
& \quad + N^{-\beta}L \int_{0}^{1}\iint_{\mathbb R^{4}}|v(x)|^{2}\sqrt{U^{\rm 3B}_{N^{\beta}}*|v|^{2}}(x) \left| \nabla \left(|v|^{2}\sqrt{U^{\rm 3B}_{N^{\beta}}*|v|^{2}}\right)(x-tN^{-\beta}y) \right| U^{\rm 3B}(y) {\rm d}x{\rm d}y{\rm d}t \\
& \leq 2\||v|^{2}\sqrt{U^{\rm 3B}_{N^{\beta}}*|v|^{2}}\|_{L^{2}}^{2} \int_{\mathbb R^{2}}\mathbbm{1}(|y|>L)U^{\rm 3B}(y){\rm d}y \\
& \quad + N^{-\beta}L \int_{0}^{1}\int_{\mathbb R^{2}} \||v|^{2}\|_{L^{5}} \left\|\sqrt{U^{\rm 3B}_{N^{\beta}}*|v|^{2}}\right\|_{L^{10}} \Big(\left\||v(\cdot-tN^{-\beta}y)|^{2}\right\|_{L^{5}} \left\|\sqrt{U^{\rm 3B}_{N^{\beta}}*|\nabla v|^{2}}(\cdot-tN^{-\beta}y)\right\|_{L^{2}} \\
& \quad + 2\left\|v(\cdot-tN^{-\beta}y)\right\|_{L^{10}} \left\|\nabla v(\cdot-tN^{-\beta}y)\right\|_{L^{2}} \left\|\sqrt{U^{\rm 3B}_{N^{\beta}}*|v|^{2}}(\cdot-tN^{-\beta}y)\right\|_{L^{10}}\Big) {\rm d}y{\rm d}t \\
& \leq 2\|v\|_{L^{6}}^{6} \int_{\mathbb R^{2}}\mathbbm{1}(|y|>L)U^{\rm 3B}(y){\rm d}y + 3N^{-\beta}L \|v\|_{L^{10}}^{5} \|\nabla v\|_{L^{2}}.
\end{align*}
Here we have used the chain rule, \Cref{lem:kinetic}, diamagnetic, H\"older and Young inequalities as well as \eqref{cv:h-nls-3-body-positive}.

The convergence rate \eqref{cv-rate:h-nls-2-body} can be found in \cite[Lemma 4.1]{LewNamRou-16}. We adapt the arguments therein to prove \eqref{cv-rate:h-nls-3-body} and \eqref{cv-rate:h-nls-3-body-modified}. We use the diamagnetic, Minkowski and Young inequalities to estimate
\begin{align*}
& \left| \int_{\mathbb R^{2}}|v(x)|^{6}{\rm d}x - \iiint_{\mathbb R^{6}} U^{\rm 3B}_{N^{\beta}}(x-y)U^{\rm 3B}_{N^{\beta}}(x-z) |v(x)|^{2} |v(y)|^{2} |v(z)|^{2} {\rm d}x {\rm d}y {\rm d}z \right| \nonumber \\
& = \left| \int_{0}^{1}\iint_{\mathbb R^{4}} |v(x)|^{2} \left[|v(x)|^{2} + (U^{\rm 3B}_{N^{\beta}}*|v|^{2})(x)\right] \nabla |v(x)|^{2}(x-t N^{-\beta}y)|^{2} \cdot N^{-\beta}yU^{\rm 3B}(y) {\rm d}x{\rm d}y{\rm d}t \right| \\
& \leq 2N^{-\beta} \int_{0}^{1}\int_{\mathbb R^{2}} \||v|^{2}\|_{L^{5}} \left( \||v|^{2}\|_{L^{5}} + \|U^{\rm 3B}_{N^{\beta}}*|v|^{2}\|_{L^{5}} \right) \times \\
& \quad \times \|v(\cdot-t N^{-\beta}y)\|_{L^{10}} \|\nabla v(\cdot-t N^{-\beta}y)\|_{L^{2}} |yU^{\rm 3B}(y)| {\rm d}y{\rm d}t \\
& \leq 4N^{-\beta} \|v\|_{L^{10}}^{5} \|\nabla v\|_{L^{2}} \int_{\mathbb R^{2}} |yU^{\rm 3B}(y)| {\rm d}y.
\end{align*}
The above is \eqref{cv-rate:h-nls-3-body}. In order to prove \eqref{cv-rate:h-nls-3-body-modified}, we use the decomposition as in the proof of \eqref{cv:h-nls-3-body-modified}. We first use the diamagnetic, Minkowski and Young inequalities to estimate
\begin{align}\label{cv-rate:h-nls-3-body-modified-1}
& \left| \int_{\mathbb R^{2}}|v(x)|^{6}{\rm d}x - \int_{\mathbb R^{2}}|v(x)|^{4}(U^{\rm 3B}_{N^{\beta}}*|v|^{2})(x) {\rm d}x \right| \nonumber \\
& = \left| \int_{0}^{1}\iint_{\mathbb R^{4}} |v(x)|^{4} \nabla |v(x)|^{2}(x-t N^{-\beta}y)|^{2} \cdot N^{-\beta}yU^{\rm 3B}(y) {\rm d}x{\rm d}y{\rm d}t \right| \nonumber \\
& \leq 2N^{-\beta} \int_{0}^{1}\int_{\mathbb R^{2}} \||v|^{4}\|_{L^{\frac{5}{2}}} \|v(\cdot-t N^{-\beta}y)\|_{L^{10}} \|\nabla v(\cdot-t N^{-\beta}y)\|_{L^{2}} |yU^{\rm 3B}(y)| {\rm d}y{\rm d}t \nonumber \\
& \leq 2N^{-\beta} \|v\|_{L^{10}}^{5} \|\nabla v\|_{L^{2}} \int_{\mathbb R^{2}} |yU^{\rm 3B}(y)| {\rm d}y.
\end{align}
Similarly,
\begin{align}\label{cv-rate:h-nls-3-body-modified-2}
& \left| \int_{\mathbb R^{2}}|v(x)|^{4}(U^{\rm 3B}_{N^{\beta}}*|v|^{2})(x) {\rm d}x - \iint_{\mathbb R^{4}}|v(x)|^{2}\sqrt{U^{\rm 3B}_{N^{\beta}}*|v|^{2}}(x)U^{\rm 3B}_{N^{\beta}}(x-y)|v(y)|^{2}\sqrt{U^{\rm 3B}_{N^{\beta}}*|v|^{2}}(z){\rm d}x{\rm d}y \right| \nonumber \\
& = \left| \int_{0}^{1}\iint_{\mathbb R^{4}} U^{\rm 3B}(y) \left| \nabla \left(|v|^{2}\sqrt{U^{\rm 3B}_{N^{\beta}}*|v|^{2}}\right)(x-tN^{-\beta}y) \right| |v(x)|^{2}\sqrt{U^{\rm 3B}_{N^{\beta}}*|v|^{2}}(x) {\rm d}x{\rm d}y{\rm d}t \right| \nonumber \\
& \leq N^{-\beta} \int_{0}^{1}\int_{\mathbb R^{2}} |yU^{\rm 3B}(y)| \Big(2\Big\|v(\cdot-tN^{-\beta}y)\Big\|_{L^{10}} \Big\|\sqrt{U^{\rm 3B}_{N^{\beta}}*|v|^{2}}(\cdot-tN^{-\beta}y)\Big\|_{L^{10}} \Big\|\nabla v(\cdot-tN^{-\beta}y)\Big\|_{L^{2}} \nonumber \\
& \quad + \Big\||v(\cdot-tN^{-\beta}y)|^{2}\Big\|_{L^{5}} \Big\|\sqrt{U^{\rm 3B}_{N^{\beta}}*|\nabla v|^{2}}(\cdot-tN^{-\beta}y)\Big\|_{L^{2}}\Big) \||v|^{2}\|_{L^{5}} \Big\|\sqrt{U^{\rm 3B}_{N^{\beta}}*|v|^{2}}\Big\|_{L^{10}} {\rm d}y{\rm d}t \nonumber \\
& \leq 3N^{-\beta} \|v\|_{L^{10}}^{5} \|\nabla v\|_{L^{2}} \int_{\mathbb R^{2}} |yU^{\rm 3B}(y)| {\rm d}y.
\end{align}
The convergence rate \eqref{cv-rate:h-nls-3-body-modified} then follows from \eqref{cv-rate:h-nls-3-body-modified-1}, \eqref{cv-rate:h-nls-3-body-modified-2} and the triangle inequality.
\end{proof}

We are now in the position to prove \Cref{thm:existence-h,thm:behaviors-h} on the existence and behaviors of Hartree ground states. In order to make use of \Cref{lem:cv-hartree-potentials}, it is worth noting that $U^{\rm 2B} \in L^{\infty}(\mathbb R^{2})$ since $\widehat{U^{\rm 2B}} \in L^{1}(\mathbb R^{2})$.

\begin{proof}[Proof of \Cref{thm:existence-h}]
The proof can be done by the concentration compactness method in the calculus of variations \cite{Lions-84a} and the properties in the NLS theory. For fixed parameters $a>a_{*}$, $b>0$, and $N>0$, let $\{v_{k}\}$ be a minimizing sequence of \eqref{energy:h}, i.e., $\|v_{k}\|_{L^{2}}^{2} = 1$ and
$$
E^{\rm H}_{a,b,N}[1] = \lim_{k\to\infty}\mathcal{E}_{a,b,N}^{\rm H}[v_{k}].
$$
For further process, we will need a property on the negativity of the Hartree energy, which is predicted by the one of the NLS. Indeed, by the variational principle, \cite[Lemma 7]{LewNamRou-17}, and \eqref{cv-rate:h-nls-3-body-0}, we have
\begin{align*}
E^{\rm H}_{a,b,N}[1] \leq{} & \mathcal{E}^{\rm H}_{a,b,N}[\ell Q_{0}(\ell\cdot)] \\ 
\leq{} & \ell^{2}\|\nabla Q_{0}\|_{L^{2}}^{2} - \ell^{2}a\|Q_{0}\|_{L^{4}}^{4} \left(\frac{1}{2} - \int_{|z|\geq L}U^{\rm 2B}(z){\rm d}z\right) \\
& + \ell^{3}aN^{-\alpha} L \|Q_{0}\|_{L^{6}}^{3} \|\nabla Q_{0}\|_{L^{2}} + \ell^{4}\frac{b}{6}\|Q_{0}\|_{L^{6}}^{6} \\
={} & \ell^{2}a\|Q_{0}\|_{L^{4}}^{4} \left(\frac{1}{2}\left(\frac{a_{*}}{a} - 1\right) - \int_{|z|\geq L}U^{\rm 2B}(z){\rm d}z\right) \\
& + \ell^{3}aN^{-\alpha} L \|Q_{0}\|_{L^{6}}^{3} \|\nabla Q_{0}\|_{L^{2}} + \ell^{4}\frac{b}{6}\|Q_{0}\|_{L^{6}}^{6}.
\end{align*}
Here $Q_{0}$ is the (unique) normalized optimizer of \eqref{ineq:gn}, and $\ell,L>0$ to be chosen. 
Now, we choose $L>0$ large enough such that
$$
\int_{|x|\geq L}U^{\rm 2B}(x){\rm d}x = 1 - \int_{|x|\leq L}U^{\rm 2B}(x){\rm d}x < \frac{1}{2}\left(1-\frac{a_{*}}{a}\right).
$$
Then the above yields that, for every $a>a_{*}$, $b>0$, and $N>0$,
\begin{equation}\label{energy:hartree-negativity}
E^{\rm H}_{a,b,N}[1] < 0
\end{equation}
provided that $\ell>0$ was chosen small enough.

Now, we show that the minimizing sequence $\{v_{k}\}$ is bounded uniformly in $H^{1}(\mathbb R^{2})$. For this purpose, we make use of an additional assumption that $\beta > \alpha$. In which case, we will show that the three-body interaction compensates for the two-body interaction. We assume on the contrary that there exits a subsequence of $\{v_{k}\}$ (still denote by $\{v_{k}\}$) such that
$$
L_{k} := \|\nabla v_{k}\|_{L^{2}} \xrightarrow{k\to\infty} \infty.
$$
Let's denote $\widetilde{v_{k}} := L_{k}^{-1}v_{k}(L_{k}^{-1}\cdot)$. By using \eqref{cv:h-nls-2-body}, we obtain
\begin{align}\label{boundedness:Hartree-1}
\mathcal{E}_{a,b,N}^{\rm H}[v_{k}] \geq{} & L_{k}^{2}\left[\int_{\mathbb R^{2}} |\nabla \widetilde{v_{k}}(x)|^{2}{\rm d}x - \frac{a}{2}\int_{\mathbb R^{2}} |\widetilde{v_{k}}(x)|^{4}{\rm d}x\right] \nonumber \\
& + L_{k}^{4}\frac{b}{6} \int_{\mathbb R^{2}} |\widetilde{v_{k}}(x)|^{2}(U^{\rm 3B}_{N^{\beta}L_{k}^{-1}}*|\widetilde{v_{k}}|^{2})(x)^{2}{\rm d}x.
\end{align}
Multiplying both sides by $L_{k}^{-2}$, neglecting the nonnegative three-body term, and taking limit $k\to\infty$, we obtain
\begin{equation}\label{boundedness:Hartree-2}
\lim_{k\to\infty}\int_{\mathbb R^{2}} |\widetilde{v_{k}}(x)|^{4}{\rm d}x \geq \frac{2}{a}\int_{\mathbb R^{2}} |\nabla \widetilde{v_{k}}(x)|^{2}{\rm d}x = \frac{2}{a} > 0.
\end{equation}
Here we have used the fact that $\|\nabla \widetilde{v_{k}}\|_{L^{2}} = 1$. On the other hand, we note that $L_{k} \leq CN^{\alpha}$, by the $L^{\infty}$-bound of the two-body interaction and the nonnegativity of the three-body interaction. If $0<\alpha<\beta$ then
$$
N^{\beta}L_{k}^{-1} \geq CL_{k}^{\frac{\beta}{\alpha}-1} \xrightarrow{k\to\infty} \infty.
$$
Multiplying both sides of \eqref{boundedness:Hartree-1} by $L_{k}^{-4}$, using \eqref{ineq:gn}, \eqref{cv:h-nls-3-body-modified}, and taking the limit $k\to\infty$, we obtain
$$
\lim_{k\to\infty}\int_{\mathbb R^{2}} |\widetilde{v_{k}}(x)|^{6}{\rm d}x = 0.
$$
By interpolation, we also have that $\widetilde{v_{k}} \to 0$ strongly in $L^{p}(\mathbb R^{2})$, for all $2 < p \leq 6$. This, however, contradicts \eqref{boundedness:Hartree-2}.

Therefore, we must have the uniform boundedness of $\{v_{k}\}$ in $H^{1}(\mathbb R^{2})$. Up to translation and extraction of a subsequence, $v_{k}$ converges to $v_{0}$ weakly in $H^{1}(\mathbb R^{2})$ and almost everywhere in $\mathbb R^{2}$. We claim that this is actually the strong convergence. We first argue that $v_{0} \not\equiv 0$ since otherwise we must have the strong convergence of $v_{k}$ to $0$ in $L^{p}(\mathbb R^{2})$, for all $2<p<\infty$. It then follows that
$$
E^{\rm H}_{a,b,N}[1] = \lim_{k\to\infty}\mathcal{E}_{a,b,N}^{\rm H}[v_{k}] \geq 0
$$
which contradicts \eqref{energy:hartree-negativity}. Assume now that $\|v_{0}\|_{L^{2}}^{2}=1$ then, by Brezis--Lieb lemma \cite{BreLie-83}, $v_{k}$ converges to $v_{0}$ strongly in $L^{2}(\mathbb R^{2})$. In fact, this strong convergence holds in $L^{p}(\mathbb R^{2})$, for all $2 \leq p<\infty$, by Sobolev embedding. It follows from Fatou's lemma that
$$
E^{\rm H}_{a,b,N}[1] = \lim_{k\to\infty}\mathcal{E}_{a,b,N}^{\rm H}[v_{k}] \geq \mathcal{E}_{a,b,N}^{\rm H}[v_{0}] \geq E^{\rm H}_{a,b,N}[1].
$$
The above yields that $v_{0}$ is a ground state of $E^{\rm H}_{a,b,N}[1]$.

It remains to eliminate the case where $0 < \lambda := \|v_{0}\|_{L^{2}}^{2} < 1$. We assume on the contrary that this is the case. Since $\{v_{k}\}$ is bounded uniformly in $H^{1}(\mathbb R^{2})$, we use \eqref{cv-rate:h-nls-2-body-0}, \eqref{cv-rate:h-nls-3-body}, and the Brezis--Lieb lemma \cite{BreLie-83} to split the energy as follows
\begin{align}\label{energy:h-split-1}
E^{\rm H}_{a,b,N}[1] = \lim_{k\to\infty}\mathcal{E}_{a,b,N}^{\rm H}[v_{k}] & \geq \lim_{k\to\infty}\mathcal{E}_{a,b}^{\rm NLS}[v_{k}] - o(1)_{N\to\infty} \nonumber \\
& = \mathcal{E}_{a,b}^{\rm NLS}[v_{0}] + \lim_{k\to\infty}\mathcal{E}_{a,b}^{\rm NLS}[v_{k}-v_{0}] - o(1)_{N\to\infty} \nonumber \\
& \geq E^{\rm NLS}_{a,b}[\lambda] + E^{\rm NLS}_{a,b}[1-\lambda] - o(1)_{N\to\infty}.
\end{align}
On the other hand, if $v_{0}$ is a ground state of $E^{\rm NLS}_{a,b}[1]$, which exists for every $a>a_{*}$ and $b>0$, then by \eqref{cv-rate:h-nls-2-body}, \eqref{cv-rate:h-nls-3-body-0}, and the variational principle,
\begin{equation}\label{energy:h-split-2}
E^{\rm H}_{a,b,N}[1] \leq \mathcal{E}_{a,b,N}^{\rm H}[v_{0}] \leq \mathcal{E}_{a,b}^{\rm NLS}[v_{0}] + o(1)_{N\to\infty} = E^{\rm NLS}_{a,b}[1] + o(1)_{N\to\infty}.
\end{equation}
Putting together \eqref{energy:h-split-1} and \eqref{energy:h-split-2}, we arrive at
\begin{equation}\label{energy:h-split-N}
o(1)_{N\to\infty} + E^{\rm NLS}_{a,b}[1] \geq E^{\rm NLS}_{a,b}[\lambda] + E^{\rm NLS}_{a,b}[1-\lambda].
\end{equation}
We shall show that this is impossible for $N$ large enough. It is worth noting that, for $a,b,\lambda>0$, either $E^{\rm NLS}_{a,b}[\lambda] = 0$ or $E^{\rm NLS}_{a,b}[\lambda] < 0$ and admits a ground state, namely $v_{\lambda}$. Furthermore, since $E^{\rm NLS}_{a,b}[1] < 0$ when $a>a_{*}$ and $b>0$, the left-hand side of \eqref{energy:h-split-N} is strictly negative when $N$ is sufficient large. In the following, we distinguish the three cases.
\begin{itemize}
\item If $E^{\rm NLS}_{a,b}[\lambda] = 0 = E^{\rm NLS}_{a,b}[1-\lambda]$ then \eqref{energy:h-split-N} obviously fail for $N$ large enough.

\item If either $E^{\rm NLS}_{a,b}[\lambda]$ or $E^{\rm NLS}_{a,b}[1-\lambda]$ vanishes, we assume without loss of generality that $E^{\rm NLS}_{a,b}[1-\lambda] = 0$, then we have
\begin{align}\label{energy:h-split-lambda}
E^{\rm NLS}_{a,b}[\lambda] = \mathcal{E}^{\rm NLS}_{a,b}[v_{\lambda}] & = \int_{\mathbb{R}^{2}}\left[|\nabla \widetilde{v_{\lambda}}|^{2} - \frac{a}{2} \|v_{\lambda}\|_{L^{2}}^{2} |\widetilde{v_{\lambda}}|^{4} + b\|v_{\lambda}\|_{L^{2}}^{2} |\widetilde{v_{\lambda}}|^{6} \right] \nonumber  \\
& = \lambda\mathcal{E}^{\rm NLS}_{a,b}[\widetilde{v_{\lambda}}] + (1-\lambda)\|\nabla \widetilde{v_{\lambda}}\|_{L^{2}}^{2} \nonumber \\
& \geq \lambda E^{\rm NLS}_{a,b}[1] + (1-\lambda)\|\nabla v_{\lambda}\|_{L^{2}}^{2} \\
& > E^{\rm NLS}_{a,b}[1]. \nonumber 
\end{align}
In this context, we define $\widetilde{v_{\lambda}} = v_{\lambda}\big(\|v_{\lambda}\|_{L^{2}}\cdot\big)$, which satisfies $\|\widetilde{v_{\lambda}}\|_{L^{2}}^{2} = 1$, and used the negativity of $E^{\rm NLS}_{a,b}[1]$. It follows again that \eqref{energy:h-split-N} fail for $N$ large enough.
\item Finally, if both $E^{\rm NLS}_{a,b}[\lambda]$ and $E^{\rm NLS}_{a,b}[1-\lambda]$ do not vanish then, by similar arguments as above,
\begin{equation}\label{energy:h-split-1-lambda}
E^{\rm NLS}_{a,b}[1-\lambda] \geq (1-\lambda)E^{\rm NLS}_{a,b}[1] + \lambda\|\nabla v_{1-\lambda}\|_{L^{2}}^{2}.
\end{equation}
We then deduce from \eqref{energy:h-split-N}, \eqref{energy:h-split-lambda} and \eqref{energy:h-split-1-lambda} that
$$
o(1)_{N\to\infty} \geq (1-\lambda)\|\nabla v_{\lambda}\|_{L^{2}}^{2} + \lambda\|\nabla v_{1-\lambda}\|_{L^{2}}^{2}.
$$
This is again a contradiction for $N$ large enough.
\end{itemize}

\end{proof}

\begin{proof}[Proof of \Cref{thm:behaviors-h}]
Firstly, we establish the cubic-quintic approximation for the Hartree problem in the limit $a_{N} \to a_{0}$ with $a_{0}>a_{*}$, as shown in \Cref{thm:behaviors-h}\ref{thm:behavior-h-cubic-quintic}. The energy upper bound in \eqref{cv:energy-h-cubic-quintic} follows from the variational principle. Taking a ground state $v_{0}$ for $E^{\rm NLS}_{a_{0},1}[1]$, using \eqref{cv:h-nls-2-body}, \eqref{cv:h-nls-3-body}, \eqref{cv:h-nls-3-body-positive-modified}, we have
\begin{align*}
E^{\rm H}_{a_{N},b_{N},N}[1] & \leq \mathcal{E}^{\rm H}_{a_{N},b_{N},N}\left[b_{N}^{-\frac{1}{2}}v_{0}\big(b_{N}^{-\frac{1}{2}}\cdot\big)\right] \\
& \leq b_{N}^{-1}\left(\mathcal{E}^{\rm NLS}_{a_{N},1}[v_{0}] + a_{N}\|v_{0}\|_{L^{4}}^{4}\int_{|z|\geq L}U^{\rm 2B}(z){\rm d}z + a_{N}b_{N}^{-\frac{1}{2}}N^{-\alpha} L \|v_{0}\|_{L^{6}}^{3} \|\nabla v_{0}\|_{L^{2}}\right).
\end{align*}
Since $a_{N} \to a_{0}$ with $a_{*}<a_{0}<\infty$, the condition $b_{N}^{-\frac{1}{2}}N^{-\alpha} \ll 1$ is therefore assumed to ensure that the error terms in the above are negligible. For instance, one can choose $L = b_{N}^{\frac{1}{4}}N^{\frac{\alpha}{2}} \gg 1$. This gives the energy upper bound in \eqref{cv:energy-h-cubic-quintic}. We establish the energy lower bound in \eqref{cv:energy-h-cubic-quintic} by proving the convergence of the sequence of Hartree ground states $\{v_{N}\}$ in \eqref{cv:gs-h-cubic-quintic}. For this purpose, we set
$$
w_{N} := b_{N}^{\frac{1}{2}} v_N \big(b_{N}^{\frac{1}{2}} \cdot \big).
$$
Then $\|w_{N}\|_{L^2} = \|v_{N}\|_{L^2} = 1$ and we rewrite the energy functional as follows
\begin{align}\label{boundedness:h-scaled-1}
b_{N}\mathcal{E}_{a_{N},b_{N},N}^{\rm H}[v_{N}]
={} & \left[\int_{\mathbb R^{2}} |\nabla w_{N}(x)|^{2}{\rm d}x - \frac{a_{N}}{2}\int_{\mathbb R^{2}} |w_{N}(x)|^{2}(U^{\rm 2B}_{N^{\alpha}b_{N}^{\frac{1}{2}}}*|w_{N}|^{2})(x){\rm d}x\right] \nonumber \\
& + \frac{1}{6} \int_{\mathbb R^{2}} |w_{N}(x)|^{2}(U^{\rm 3B}_{N^{\beta}b_{N}^{\frac{1}{2}}}*|w_{N}|^{2})(x)^{2}{\rm d}x.
\end{align}
The left-hand side of the above is bounded uniformly, by the energy upper bound in \eqref{cv:energy-h-cubic-quintic}. Then, under the assumptions that $\beta>\alpha$ and that $a_{N} \to a_{0}$ with $a_{*}<a_{0}<\infty$, we can prove that $\{w_{N}\}$ is bounded uniformly in $H^{1}(\mathbb R^{2})$, by the same arguments as in the proof of \Cref{thm:behaviors-h}. Up to translation and extraction of a subsequence, $w_{N}$ converges to $w_{0}$ weakly in $H^{1}(\mathbb R^{2})$ and almost everywhere in $\mathbb R^{2}$. We claim that this is actually the strong convergence. We first argue that $w_{0} \not\equiv 0$ since otherwise we must have the strong convergence of $w_{N}$ to $0$ in $L^{p}(\mathbb R^{2})$, for all $2<p<\infty$. Taking the limit $N\to\infty$ in \eqref{boundedness:h-scaled-1}, using \Cref{lem:cv-hartree-potentials} and the energy upper bound in \eqref{cv:energy-h-cubic-quintic}, we obtain
$$
E_{a_{0},1}^{\rm NLS}[1] \geq \lim_{N\to\infty}b_{N}\mathcal{E}_{a_{N},b_{N},N}^{\rm H}[v_{N}] \geq 0.
$$
This, however, contradicts the negatively of $E_{a_{0},1}^{\rm NLS}[1]$ for $a_{0}>a_{*}$ (see \cite{DoaNgu-26}). Therefore, $w_{0} \not\equiv 0$. Next, we eliminate the case where $0 < \lambda := \|w_{0}\|_{L^{2}}^{2} < 1$. We assume on the contrary that this is the case. Then taking again the limit $N\to\infty$ in \eqref{boundedness:h-scaled-1}, using \Cref{lem:cv-hartree-potentials}, Brezis--Lieb lemma \cite{BreLie-83}, and the energy upper bound in \eqref{cv:energy-h-cubic-quintic}, we split the energy as follows
\begin{align}\label{energy:h-split}
E_{a_{0},1}^{\rm NLS}[1] \geq \lim_{N\to\infty}b_{N}\mathcal{E}_{a_{N},b_{N},N}^{\rm H}[v_{N}] & \geq \lim_{N\to\infty}\mathcal{E}_{a_{0},1}^{\rm NLS}[w_{N}] \nonumber \\
& = \mathcal{E}_{a_{0},1}^{\rm NLS}[w_{0}] + \lim_{N\to\infty}\mathcal{E}_{a_{0},1}^{\rm NLS}[w_{N}-w_{0}] \nonumber \\
& \geq E_{a_{0},1}^{\rm NLS}[\lambda] + E_{a_{0},1}^{\rm NLS}[1-\lambda].
\end{align}
This, however, contradicts the strict binding inequality
$$
E_{a_{0},1}^{\rm NLS}[1] < E_{a_{0},1}^{\rm NLS}[\lambda] + E_{a_{0},1}^{\rm NLS}[1-\lambda]
$$
for every $a_{0}>a_{*}$ and $0<\lambda<1$, which can be proved by similar (even simpler) arguments in the proof of \Cref{thm:existence-h}. Hence, we have proved that we must have that $\|w_{0}\|_{L^{2}}^{2}=1$. Then, by Brezis--Lieb lemma \cite{BreLie-83}, $w_{N}$ converges to $w_{0}$ strongly in $L^{2}(\mathbb R^{2})$. In fact, this strong convergence holds in $L^{p}(\mathbb R^{2})$, for all $2 \leq p<\infty$, by Sobolev embedding. Again taking the limit $N\to\infty$ in \eqref{boundedness:h-scaled-1}, using \Cref{lem:cv-hartree-potentials}, Fatou's lemma, and the energy upper bound in \eqref{cv:energy-h-cubic-quintic}, we obtain
$$
E_{a_{0},1}^{\rm NLS}[1] \geq \lim_{N\to\infty}b_{N}\mathcal{E}_{a_{N},b_{N},N}^{\rm H}[v_{N}] \geq \lim_{N\to\infty}\mathcal{E}_{a_{0},1}^{\rm NLS}[w_{N}] \geq \mathcal{E}_{a_{0},1}^{\rm NLS}[w_{0}] \geq E_{a_{0},1}^{\rm NLS}[1].
$$
The above yields \eqref{cv:energy-h-cubic-quintic} as well as the convergence $\|\nabla w_{N}\|_{L^{2}}^{2} \to \|\nabla w_{0}\|_{L^{2}}^{2}$. Hence, $w_{N} \to w_{0}$ strongly in $H^{1}(\mathbb R^{2})$, by Br\'ezis--Lieb lemma \cite{BreLie-83}. This completes the proof of \eqref{cv:gs-h-cubic-quintic} and \Cref{thm:behaviors-h}\ref{thm:behavior-h-cubic-quintic} as well.

Secondly, we establish the cubic approximation for the Hartree problem in the limit $a_{N} \to a_{*}$, as shown in \Cref{thm:behaviors-h}\ref{thm:behavior-h-cubic}. Let $Q_{0}$ be the unique $L^{2}$-normalized optimizer for \eqref{ineq:gn} and $\ell_{N}$ be as in \eqref{cv:gs-h-cubic}. By \eqref{cv-rate:h-nls-2-body}, \eqref{cv-rate:h-nls-3-body-0}, and the variational principle,
\begin{align}\label{cv:energy-cubic-2d-upper-bound}
E^{\rm H}_{a_{N},b_{N},N}[1] & \leq \mathcal{E}_{a_{N},b_{N},N}^{\rm H}\left[\ell_{N}^{-\frac{1}{2}} Q_{0}\big(\ell_{N}^{-\frac{1}{2}} \cdot\big)\right] \nonumber \\
& \leq -\ell_{N}^{-1}\frac{a_{N} - a_{*}}{2} \|Q_{0}\|_{L^{4}}^{4} + \ell_{N}^{-2}\frac{b_{N}}{6} \|Q_{0}\|_{L^{6}}^{6} + a_{N}N^{-\alpha}\ell_{N}^{-\frac{3}{2}} \|Q_{0}\|_{L^{6}}^{3} \|\nabla Q_{0}\|_{L^{2}} \int_{\mathbb R^{2}} \left|xU^{\rm 2B}(x)\right| {\rm d}x\nonumber \\
& = \frac{(a_{N} - a_{*})^{2}}{b_{N}} \left(-\frac{3\|Q_{0}\|_{L^{4}}^{8}}{8\|Q_{0}\|_{L^{6}}^{6}} + CN^{-\alpha}(a_{N}-a_{*})^{-\frac{1}{2}}b_{N}^{-\frac{1}{2}}\right).
\end{align}
Since $a_{N} \to a_{*}$, the condition $N^{-\alpha}(a_{N}-a_{*})^{-\frac{1}{2}}b_{N}^{-\frac{1}{2}} \ll 1$ is therefore assumed to ensure that the error term in the above is negligible. This gives the energy upper bound in \eqref{cv:energy-h-cubic}. We prove the matching energy lower bound in \eqref{cv:energy-h-cubic} by showing the convergence of ground states in \eqref{cv:gs-h-cubic}. For this purpose, we set
$$
w_{N}(x) := \ell_{N}^{\frac{1}{2}} v_{N}(\ell_{N}^{\frac{1}{2}} x)
$$
where $\ell_{N}$ is given in \eqref{cv:gs-h-cubic}. Then $\|w_{N}\|_{L^2} = \|v_{N}\|_{L^2} = 1$ and we rewrite the energy functional as follows
\begin{align}
\ell_{N}^{2}b_{N}^{-1}\mathcal{E}^{\rm H}_{a_{N}, b_{N},N} [v_{N}] ={} & \ell_{N}b_{N}^{-1}\left[\int_{\mathbb R^{2}} |\nabla w_{N}(x)|^{2}{\rm d}x - \frac{a_{N}}{2}\int_{\mathbb R^{2}} |w_{N}(x)|^{2}(U^{\rm 2B}_{N^{\alpha}\ell_{N}^{\frac{1}{2}}}*|w_{N}|^{2})(x){\rm d}x\right] \nonumber \\
& + \frac{1}{6} \int_{\mathbb R^{2}} |w_{N}(x)|^{2}(U^{\rm 3B}_{N^{\beta}\ell_{N}^{\frac{1}{2}}}*|w_{N}|^{2})(x)^{2}{\rm d}x.\label{cv:energy-cubic-2d-lower-bound-2}
\end{align}
The left-hand side of the above is bounded uniformly, by the energy upper bound in \eqref{cv:energy-h-cubic}. Then, under the assumptions that $\beta>\alpha$ and that $a_{N} \to a_{*}$, we can prove that $\{w_{N}\}$ is bounded uniformly in $H^{1}(\mathbb R^{2})$, by the same arguments as in the proof of \Cref{thm:existence-h}. 

It follows from the energy upper bound in \eqref{cv:energy-h-cubic}, \eqref{ineq:gn}, and \eqref{cv:energy-cubic-2d-lower-bound-2} that
\begin{equation}\label{cv:energy-cubic-2d-boundedness}
-1 +o(1)_{N\to\infty} \geq -2\|Q_{0}\|_{L^{4}}^{-4} \|w_{N}\|_{L^{4}}^{4} + \|Q_{0}\|_{L^{6}}^{-6} \|w_{N}\|_{L^{6}}^{6}.
\end{equation}
Since $\|w_{N}\|_{L^{2}}^{2}=1$ and $\|w_{N}\|_{L^{4}}^{4} \leq \|w_{N}\|_{L^{6}}^{3}$, by H\"older inequality, we deduce from \eqref{cv:energy-cubic-2d-boundedness} that $\|w_{N}\|_{L^{4}}$ as well as $\|w_{N}\|_{L^{6}}$ are bounded uniformly from above and below. This also yields the uniform boundedness of $\|\nabla w_{N}\|_{L^{2}}$. Indeed, looking back at \eqref{cv:energy-cubic-2d-lower-bound-2}, by multiplying both sides by $a_{N}-a_{*}$, neglecting the nonnegative three-body interaction, using again the energy upper bound in \eqref{cv:energy-h-cubic}, and taking the limit $N\to\infty$, we obtain
\begin{equation}\label{cv:gs-cubic-2d-profile}
0 \geq \lim_{N\to\infty}\|\nabla w_{N}\|_{L^{2}}^{2} - \frac{a_{N}}{2} \|w_{N}\|_{L^{4}}^{4} = \lim_{N\to\infty} \|\nabla w_{N}\|_{L^{2}}^{2} - \frac{a_{*}}{2} \|w_{N}\|_{L^{4}}^{4} \geq 0.
\end{equation}
Here we have used \eqref{ineq:gn} in the last inequality. Therefore, the equality in \eqref{cv:gs-cubic-2d-profile} must occur and this implies in particular the uniform boundedness of $\{w_{N}\}$ in $H^1(\mathbb{R}^{2})$. Up to translation and extraction of a subsequence, $w_{N}$ converges to $w_{0}$ weakly in $H^{1}(\mathbb R^{2})$ and almost everywhere in $\mathbb R^{2}$. We claim that $w_{0} \not\equiv 0$ since otherwise we must have the strong convergence of $w_{N}$ to $0$ in $L^{p}(\mathbb R^{2})$, for all $2<p<\infty$. This, however, contradicts \eqref{cv:energy-cubic-2d-boundedness}. Now, looking back at \eqref{cv:gs-cubic-2d-profile}, we use Brezis--Lieb lemma \cite{BreLie-83} and \eqref{ineq:gn} to obtain
\begin{align*}
0 & = \lim_{N\to\infty} \|\nabla w_{N}\|_{L^{2}}^{2} - \frac{a_{*}}{2} \|w_{N}\|_{L^{4}}^{4} \\
& \geq \|\nabla w_{0}\|_{L^{2}}^{2} - \frac{a_{*}}{2} \|w_{0}\|_{L^{4}}^{4} + \lim_{N\to\infty} \|\nabla (w_{N}-w_{0})\|_{L^{2}}^{2} - \frac{a_{*}}{2} \|w_{N}-w_{0}\|_{L^{4}}^{4} \\
& \geq \left(1-\|w_{0}\|_{L^{2}}^{2}\right) \|\nabla w_{0}\|_{L^{2}}^{2} \\
& \geq 0.
\end{align*}
Here we have utilized the facts that $0 < \|w_{0}\|_{L^{2}}^{2} \leq 1$. Therefore, the equality must occur in the above and we must have that $\|w_{0}\|_{L^{2}}^{2} = 1$. By Br\'ezis--Lieb lemma \cite{BreLie-83}, $w_{N} \to w_{0}$ strongly in $L^{2}(\mathbb R^{2})$. By Sobolev embedding, this strong convergence holds in $L^{p}(\mathbb R^{2})$, for all $2 \leq p<\infty$. Then the above yields the convergence $\|\nabla w_{N}\|_{L^{2}}^{2} \to \|\nabla w_{0}\|_{L^{2}}^{2}$. Hence, $w_{N} \to w_{0}$ strongly in $H^{1}(\mathbb R^{2})$, again by Br\'ezis--Lieb lemma \cite{BreLie-83}. Furthermore, it also yields that $w_{0}$ is an optimizer for \eqref{ineq:gn}, which is unique (up to translation and dilation). Up to a translation, $w_{0}(x) = \sqrt{t}Q_0(\sqrt{t}x)$ where $t>0$ and $Q_{0}$ is the unique $L^{2}$-normalized optimizer for \eqref{ineq:gn}. To complete the proof, we prove that $t = 1$ and hence $w_{0} \equiv Q_0$. Indeed, by taking the limit $N\to \infty$ in \eqref{cv:energy-cubic-2d-boundedness}, we obtain
$$
-1 \geq -2\|Q_{0}\|_{L^{4}}^{-4} \|w_{0}\|_{L^{4}}^{4} + \|Q_{0}\|_{L^{6}}^{-6} \|w_{0}\|_{L^{6}}^{6} = -2t +t^{2} \geq -1. 
$$
Here we have used the Cauchy--Schwarz inequality in the last inequality. Obviously, the equality in the above occurs at $t=1$. Finally, the convergence \eqref{cv:gs-cubic-2d} holds for the whole sequence since the limiting profile $Q_{0}$ is unique. The proof is completed.

Finally, we establish the TF approximation for the Hartree problem in the limit $a_{N} \to \infty$, as shown in \Cref{thm:behaviors-h}\ref{thm:behavior-h-tf}. In order to estimate the energy upper bound in \eqref{cv:energy-h-tf}, we follow the arguments in \cite{DoaNgu-26} and we need to further control the error term made of the two-body interaction. Let $g_{\mu} = (2\pi)^{-1}\mu e^{-\sqrt{\mu} |x|}$ with $\mu>0$ to be chosen, and $\varrho^{\rm TF}_{1,1}$ denotes the (unique) ground state of $E^{\rm TF}_{1,1}[1]$. We consider 
$$
v_{N} = \sqrt{g_{\mu}*\varrho^{\rm TF}_{N}} \quad \text{with} \quad \varrho^{\rm TF}_{N}(x) = \frac{a_{N}}{b_{N}}\varrho^{\rm TF}_{1,1}\left(\left(\frac{a_{N}}{b_{N}}\right)^{\frac{1}{2}}x\right)
$$
as the trial state for $E^{\rm H}_{a_{N},b_{N},N}[1]$. It is worth noting that $\int_{\mathbb R^{2}}g_{\mu} = 1$ and $|\nabla g_{\mu}| = \sqrt{\mu} g_{\mu}$, for all $\mu>0$, which implies that
\begin{equation}\label{cv:energy-h-tf-kinetic}
|\nabla v_{N}| = \frac{|\nabla g_{\mu}*\varrho^{\rm TF}_{N}|}{2\sqrt{g_{\mu}*\varrho^{\rm TF}_{N}}} \leq \frac{\sqrt{\mu}}{2}\sqrt{g_{\mu}*\varrho^{\rm TF}_{N}}.
\end{equation}
By the variational principle, \cite[Lemma 7]{LewNamRou-17}, \eqref{cv-rate:h-nls-2-body}, \eqref{cv-rate:h-nls-3-body-0}, \eqref{cv:energy-h-tf-kinetic}, and Young inequality, we have
\begin{align*}
E^{\rm H}_{a_{N},b_{N},N}[1] \leq{} & \mathcal{E}^{\rm H}_{a_{N},b_{N},N}[v_{N}] \\
\leq{} & \|\nabla v_{N}\|_{L^{2}}^{2} + a_{N}N^{-\alpha}L\|\nabla v_{N}\|_{L^{2}}\|v_{N}\|_{L^{6}}^{3} \\
& - \left(1-2\int_{|z|\geq L}U^{\rm 2B}(z){\rm d}z\right)\frac{a_{N}}{2} \int_{\mathbb R^{2}}|v_{N}(x)|^{4}{\rm d}x + \frac{b_{N}}{6} \int_{\mathbb R^{2}}|v_{N}(x)|^{6}{\rm d}x \\
\leq{} & \frac{\mu}{4} + \frac{a_{N}}{2}\sqrt{\mu}N^{-\alpha}L\|\varrho^{\rm TF}_{N}\|_{L^{3}}^{\frac{3}{2}} + (1-o(1)_{N\to\infty})\frac{a_{N}}{2}\left(\int_{\mathbb R^{2}}\varrho^{\rm TF}_{N}(x)^{2} - \int_{\mathbb R^{2}}(g_{\mu}*\varrho^{\rm TF}_{N})(x)^{2}\right) \\
& - (1-o(1)_{N\to\infty})\frac{a_{N}}{2} \int_{\mathbb R^{2}}\varrho^{\rm TF}_{N}(x)^{2}{\rm d}x + \frac{b_{N}}{6} \int_{\mathbb R^{2}}\varrho^{\rm TF}_{N}(x)^{3}{\rm d}x \\
={} & \frac{\mu}{4} + \frac{a_{N}^{2}}{2b_{N}}\sqrt{\mu}N^{-\alpha}L\|\varrho^{\rm TF}_{1,1}\|_{L^{3}}^{\frac{3}{2}} \\
& + (1-o(1)_{N\to\infty})\frac{a_{N}^{2}}{2b_{N}}\left(\int_{\mathbb R^{2}}\varrho^{\rm TF}_{1,1}(x)^{2} - \int_{\mathbb R^{2}}\left(g_{\mu \frac{b_{N}}{a_{N}}}*\varrho^{\rm TF}_{1,1}\right)(x)^{2}\right) \\
& + (1+o(1)_{N\to\infty})E^{\rm TF}_{a_{N},b_{N}}.
\end{align*}
Here $L=L_{N} \gg 1$ to be chosen. Now, we optimize over $\mu>0$ the first two terms on the right hand side of the above, and we choose $L$ in such a way that the optimal $\mu$ satisfies
$$
\mu = \mu_{N} = \mathcal{O}\left(\left(\frac{a_{N}^{2}}{b_{N}}N^{-\alpha}L\right)^{2}\right) = o\left(E^{\rm TF}_{a_{N},b_{N}}[1]\right)_{N\to\infty} = o\left(\frac{a_{N}^{2}}{b_{N}}\right)_{N\to\infty} \quad \text{and} \quad \lim_{N\to\infty}\mu_{N} \frac{b_{N}}{a_{N}} = \infty.
$$
This can be done under the assumption that $\dfrac{a_{N}^{2}}{b_{N}} \ll N^{2\alpha}$. For instance, we can take
$$
L = L_{N }= \left(\frac{b_{N}}{a_{N}^{2}}N^{2\alpha}\right)^{\frac{1}{2}} \left(\min\left\{\left(\frac{b_{N}}{a_{N}^{2}}N^{2\alpha}\right)^{\frac{1}{4}},a_{N}^{\frac{1}{2}}\right\}\right)^{-1} \gg 1.
$$
Consequently, it follows from the above that (see, e.g., \cite[Theorem 2.16]{LieLos-01})
$$
\lim_{N\to\infty}g_{\mu \frac{b_{N}}{a_{N}}}*\varrho^{\rm TF}_{1,1} = \varrho^{\rm TF}_{1,1}
$$
strongly in $L^{2}(\mathbb R^{2})$. Putting all together, we obtain the desired energy upper bound in \eqref{cv:energy-h-tf}.

We estimate the matching energy lower bound in \eqref{cv:energy-h-tf} by proving the convergence of ground states in \eqref{cv:gs-h-tf}. For this purpose, we need to control the error term made of the three-body interaction. Let $\{v_{N}\}$ be a sequence of ground states of $E^{\rm H}_{a_{N},b_{N},N}[1]$, for $a_{N} > a_{*}$, $b_{N} > 0$, and $N$ large enough. Setting 
$$
\varrho_{N}(x) = \frac{b_{N}}{a_{N}} v_{N}\left(\left(\frac{b_{N}}{a_{N}}\right)^{\frac{1}{2}}x\right)^{2}.
$$
Using \eqref{cv-rate:h-nls-2-body-0}, \eqref{cv-rate:h-nls-3-body}, and the nonnegativity of the kinetic term, we estimate
\begin{align*}
E^{\rm H}_{a_{N},b_{N},N}[1] ={} & \mathcal{E}^{\rm H}_{a_{N},b_{N},N}[v_{N}] \\
\geq{} & -\frac{a_{N}}{2} \int_{\mathbb R^{2}}|v_{N}(x)|^{4}{\rm d}x + \left(1-4\int_{|z|\geq L}U^{\rm 3B}(z){\rm d}z\right)\frac{b_{N}}{6} \int_{\mathbb R^{2}}|v_{N}(x)|^{6}{\rm d}x \\
& - \frac{2}{3}b_{N}N^{-\beta}L \|v_{N}\|_{L^{10}}^{5} \|\nabla v_{N}\|_{L^{2}} \\
={} & \frac{a_{N}^{2}}{b_{N}} \left[-\frac{1}{2}\int_{\mathbb R^{2}}\varrho_{N}(x)^{2}{\rm d}x + \frac{1-o(1)_{N\to\infty}}{6}\int_{\mathbb R^{2}}\varrho_{N}(x)^{3}{\rm d}x - \frac{2b_{N}^{2}}{3a_{N}^{2}}N^{5\alpha-\beta}L\right].
\end{align*}
Here $L=L_{N} \gg 1$ to be chosen, and we have used the fact that 
$$
\|v_{N}\|_{L^{10}}^{\frac{5}{4}} \leq C\|\nabla v_{N}\|_{L^{2}} \leq CN^{\alpha},
$$
by the Sobolev embedding in $H^{1}(\mathbb R^{2})$, the $L^{\infty}$-bound of the two-body interaction, the nonnegativity of the three-body interaction, and the energy upper bound in \eqref{cv:energy-h-tf} as well as the condition $\dfrac{a_{N}^{2}}{b_{N}} \ll N^{2\alpha}$. Under the further assumption that $\dfrac{b_{N}^{2}}{a_{N}^{2}}N^{5\alpha-\beta} \ll 1$, we can choose $L$ in such a way that the error term in the above is of order $1$. Together with the energy upper bound in \eqref{cv:energy-h-tf}, we deduce that
\begin{equation}\label{cv:energy-TF-lower-bound}
o(1)_{N\to\infty} + E^{\rm TF}_{1,1}[1] \geq -\frac{1}{2}\int_{\mathbb R^{2}}\varrho_{N}(x)^{2}{\rm d}x + \frac{1-o(1)_{N\to\infty}}{6}\int_{\mathbb R^{2}}\varrho_{N}(x)^{3}{\rm d}x.
\end{equation}
By the H\"older inequality, \eqref{cv:energy-TF-lower-bound} yields that $\{\varrho_{N}\}_{N}$ is bounded uniformly in $L^{1}\cap L^{3}(\mathbb R^{2})$. Unfortunately, the compactness of $\{\varrho_{N}\}_{N}$ is not obvious, due to the locality of the TF functional. But the one of its rearrangement, namely $\varrho_{N}^{*}$, can be recovered by using its radially symmetric decreasing property. This nice property was used in the studies of the homogeneous cubic-quintic NLS and TF theory \cite{DoaNgu-26}. Because of the norm-preserving, i.e., $\|\varrho_{N}\|_{L^{p}} = \|\varrho_{N}^{*}\|_{L^{p}}$, for all $p \geq 1$, we have that $\{\varrho_{N}^{*}\}_{N}$ is still bounded uniformly in $L^{1}\cap L^{3}(\mathbb R^{2})$. Up to translation and extracting a subsequence, $\varrho_{N}^{*} \to \varrho_{0}$ weakly in $L^{1}\cap L^{3}(\mathbb R^{2})$ and pointwise almost everywhere in $\mathbb R^{2}$. Since, $\{\varrho_{N}^{*}\}_{N}$ is radially symmetric decreasing, we can prove that $\varrho_{N}^{*} \to \varrho_{0}$ strongly in $L^{p}(\mathbb R^{2})$, for all $1 < p < 3$, by the arguments in \cite{DoaNgu-26}. We prove that this convergence also holds true in $L^{1}(\mathbb R^{2})$ and $L^{3}(\mathbb R^{2})$ as well.

We first note that the limiting profile $\varrho_{0}$ is non-trivial since otherwise we must have that $E^{\rm TF}_{1,1}[1] \geq 0$, by taking the limit $N \to \infty$ in \eqref{cv:energy-TF-lower-bound}, after replacing $\varrho_{N}$ by $\varrho_{N}^{*}$. This contradicts the negativity of $E^{\rm TF}_{1,1}[1]$. On the other hand, taking again the limit $N \to \infty$ in \eqref{cv:energy-TF-lower-bound}, after replacing again $\varrho_{N}$ by $\varrho_{N}^{*}$, using Fatou's lemma and the weak convergence $\varrho_{N}^{*} \wto \varrho_{0}$ in $L^{3}(\mathbb R^{2})$ as well as the strong convergence $\varrho_{N}^{*} \to \varrho_{0}$ in $L^{2}(\mathbb R^{2})$, we obtain
$$
E^{\rm TF}_{1,1}[1] \geq \mathcal{E}^{\rm TF}_{1,1}[\varrho_{0}] = \mathcal{E}^{\rm TF}_{1,1}[\widetilde{\varrho_{0}}]\int_{\mathbb R^{2}}\varrho_{0}(x){\rm d}x \geq E^{\rm TF}_{1,1}[1]\int_{\mathbb R^{2}}\varrho_{0}(x){\rm d}x \geq E^{\rm TF}_{1,1}[1].
$$
In this context, we define $\widetilde{\varrho_{0}}(x) = \varrho_{0}\left(\left(\int_{\mathbb R^{2}}\varrho_{0}\right)^{\frac{1}{2}}x\right)$, which satisfies $\int_{\mathbb R^{2}}\widetilde{\varrho_{0}} = 1$. We have utilized the facts that $0 < \int_{\mathbb R^{2}}\varrho_{0} \leq 1$ and that $E^{\rm TF}_{1,1}[1] < 0$. The equality must hold in the above expression, implying that $\int_{\mathbb R^{2}}\varrho_{0} = 1$ and that $\int_{\mathbb R^{2}}(\varrho_{N}^{*})^{3} \to \int_{\mathbb R^{2}}\varrho_{0}^{3}$. By the Br\'ezis--Lieb lemma \cite{BreLie-83}, $\varrho_{N}^{*} \to \varrho_{0}$ strongly in $L^{1} \cap L^{3}(\mathbb R^{2})$. This also establishes the energy lower bound in \eqref{cv:energy-h-tf} and confirms that $\varrho_{0}$ is a ground state of $E^{\rm TF}_{1,1}[1]$. Finally, the convergence \eqref{cv:gs-h-tf} holds for the whole sequence since the limiting profile is unique. The proof is completed.

\end{proof}

\section{Asymptotic behavior of the quantum energy}\label{sec:many-body}

The objective of this section is to demonstrate Theorem \ref{thm:qm}. Given the established connection between the intermediate Hartree theory and the semiclassical NLS theory, it remains challenging to compare the quantum energy $E_{a,b,N}^{\rm QM}$ expressed in \eqref{energy:quantum} with the Hartree energy $E^{\rm H}_{a,b,N}[1]$ defined in \eqref{energy:h}. For a system of bosons at very low temperature, it is expected that all particles occupy the same quantum state. The mathematical expression of this phenomenon is that the many-body wave function is approximated by the same one-body wave function for all variables, i.e., \eqref{eq:BEC}. It is then natural to use \eqref{eq:BEC} as a trial state to derive an energy upper bound. By the calculus of variations and a simple calculation, we have
$$
E_{a,b,N}^{\mathrm{QM}} \leq \left\langle v^{\otimes N} \left|\frac{H_{a,b,N}}{N}\right| v^{\otimes N} \right\rangle = E^{\rm H}_{a,b,N}[1].
$$
The estimation of the energy lower bound presents a technical obstacle. Consequently, direct calculation of this lower bound is not feasible, necessitating the consideration of the modified Hartree theory, as outlined in \eqref{energy:h-modified}-\eqref{functional:h-modified} below. We have the following.

\begin{theorem}\label{thm:qm-h}
Under the assumptions \eqref{scaled-interaction}, \eqref{condition:2-3-body} and assuming further that $\widehat{U^{\rm 3B}} \geq 0$, we have
\begin{equation}\label{cv:energy-qm-hartree}
E^{\rm H}_{a,b,N}[1] \geq E^{\rm QM}_{a,b,N} \geq E^{\rm mH}_{a,b,N} - C\left(a_{N}N^{2\alpha-1}+b_{N}N^{4\beta-1}\right).
\end{equation}
Here $E^{\rm mH}_{a,b,N}$ is the modified Hartree energy, given by \eqref{energy:h-modified}.
\end{theorem}

The proof of Theorem \ref{thm:qm-h} comprises three key components. First, we employ the Hoffmann--Ostenhof inequality \cite{Hof-77} to directly estimate the kinetic energy, i.e.,
\begin{equation}\label{ineq:HO}
\tr\left[-\Delta\gamma_{\Psi_N}^{(1)}\right] \geq \int_{\mathbb R^2}\left|\nabla\sqrt{\varrho_{\gamma_{\Psi_N}^{(1)}}}(x)\right|^{2}{\rm d}x
\end{equation}
for any wave function $\Psi_{N}$ belonging to $L^2_{\rm sym}(\mathbb R^{2N})$. Here $\gamma_{\Psi_{N}}^{(1)}$ is the $1$-particle reduced density matrices, defined by the partial trace
$$
\gamma_{\Psi_{N}}^{(1)} := \tr_{2\to N} | \Psi_{N} \rangle \langle \Psi_{N} |
$$
and $\varrho_{\gamma_{\Psi_{N}}^{(1)}}(x) = \gamma_{\Psi_{N}}^{(1)}(x,x)$ is its diagonal. Equivalently, $\gamma_{\Psi_{N}}^{(1)}$ is the trace class self-adjoint operator on $L^{2}(\mathbb R^{2})$ with the kernel
$$
\gamma_{\Psi_{N}}^{(1)}(x,y) := \int_{\mathbb R^{2(N-1)}}\overline{\Psi_{N}(x,Z)}{\Psi_{N}(y,Z)} {\rm d} Z .
$$
A proof of \eqref{ineq:HO} utilizing the convexity of the kinetic energy \cite[Theorem 7.8]{LieLos-01} is provided in \cite[Lemma 3.2]{Lewin-ICMP}.
Second, the two-body interaction potential $U^{\rm 2B}$ can be lower-bounded by a one-body cubic term, employing Onsager’s lemma \cite{Onsager-39}. Such a lemma is a classical idea used to control Coulomb-type or positive-definiteness two-body interactions by smearing point charges with a Radon measure. It is widely used in stability-of-matter proofs developed by Lieb and collaborators (see e.g., \cite{LieSei-10}). A simple argument, under the assumption of positive definiteness of $U^{\rm 2B}$, i.e., $\widehat{U^{\rm 2B}} \geq 0$, is by the observation
$$
\int_{\mathbb R^{2}}f(x)U^{\rm 2B}(x-y)f(y){\rm d}x{\rm d}y = \int_{\mathbb R^{2}}\widehat{U^{\rm 2B}}(k)\left|\widehat{f}(k)\right|^{2} \geq 0
$$
applied to $\displaystyle f = \sum_{i=1}^{N}\delta_{x_{i}} - N\varrho_{\gamma_{\Psi_{N}}^{(1)}}$. In order to treat more general two-body interactions, a technique devised by Lévy--Leblond \cite{LevLeb-69} was used. A comprehensive discussion of this approach is presented in \cite[Section 3]{Lewin-ICMP} (see also \cite{Lewin-ICM} and \cite{Rougerie-20}). Third, we adapt this strategy to derive a modified version of Onsager’s lemma tailored to the repulsive three-body interaction, showing that it can be controlled by a one-body quintic term. The factorized structure of the three-body interaction in \eqref{hamiltonian} is assumed for technical reasons. In particular, it enables us to rewrite certain three-body contributions as products of two-body convolutions, which play a crucial role in the estimates used to control the interaction terms. This reduction is essential in our argument and does not appear to extend to the original three-body interaction \eqref{3-body-interaction}. We thus work with the factorized model throughout the analysis. The corresponding result is stated in the following.

\begin{lemma}\label{lem:three-one}
If $0 \leq U^{\rm 3B}$ and $0 \leq \widehat{U^{\rm 3B}} \in L^{1}(\mathbb R^{2})$ then, for any nonnegative integrable function $\chi \geq 0$, we have
\begin{align}\label{ineq:three-one}
\sum_{1\leq i\ne j\ne k \leq N}U^{\rm 3B}(x_{i}-x_{j})U^{\rm 3B}(x_{i}-x_{k})  \geq{} & 2\sum_{1\leq i\leq N} \sqrt{U^{\rm 3B}*\chi}(x_{i}) (U^{\rm 3B}*\chi\sqrt{U^{\rm 3B}*\chi})(x_{i})\nonumber \\
& - \iint_{\mathbb R^{4}}\chi(y)\sqrt{U^{\rm 3B}*\chi}(y)U^{\rm 3B}(y-z)\chi(z)\sqrt{U^{\rm 3B}*\chi}(z){\rm d}y{\rm d}z
\nonumber \\
& - \sum_{1\leq i,j \leq N}U^{\rm 3B}(x_{i}-x_{j})^{2} - 2U^{\rm 3B}(0)\sum_{1\leq i\ne j \leq N}U^{\rm 3B}(x_{i}-x_{j}).
\end{align}
\end{lemma}

\begin{proof} 
Let $\displaystyle f = \sum_{i=1}^{N}\delta_{x_{i}} - \chi$ be a Radon measure. We have
$$
\sum_{1\leq i \leq N}\iint_{\mathbb R^{4}}f(y)U^{\rm 3B}(x_{i}-y)U^{\rm 3B}(x_{i}-z)f(z){\rm d}y{\rm d}z = \sum_{1\leq i \leq N}(f*U^{\rm 3B})(x_{i})^{2} \geq 0.
$$
This is equivalent to
\begin{align}\label{ineq:three-one-1}
\sum_{1\leq i\ne j\ne k \leq N}U^{\rm 3B}(x_{i}-x_{j})U^{\rm 3B}(x_{i}-x_{k}) \geq{} & 2\sum_{1\leq i,j \leq N}U^{\rm 3B}(x_{i}-x_{j})(U^{\rm 3B}*\chi)(x_{i}) - \sum_{1\leq i \leq N}(U^{\rm 3B}*\chi)(x_{i})^{2} \nonumber \\
& - \sum_{1\leq i,j \leq N}U^{\rm 3B}(x_{i}-x_{j})^{2} - 2U^{\rm 3B}(0)\sum_{1\leq i\ne j \leq N}U^{\rm 3B}(x_{i}-x_{j}).
\end{align}
In order to estimate the main term in \eqref{ineq:three-one-1}, we use the symmetry of $U^{\rm 3B}$ to get
\begin{align}\label{ineq:three-one-2}
\sum_{1\leq i,j \leq N}U^{\rm 3B}(x_{i}-x_{j})(U^{\rm 3B}*\chi)(x_{i}) & = \sum_{1\leq i,j \leq N}U^{\rm 3B}(x_{i}-x_{j})(U^{\rm 3B}*\chi)(x_{j}) \nonumber \\
& = \sum_{1\leq i,j \leq N}U^{\rm 3B}(x_{i}-x_{j})\frac{(U^{\rm 3B}*\chi)(x_{i})+(U^{\rm 3B}*\chi)(x_{j})}{2} \nonumber \\
& \geq \sum_{1\leq i,j \leq N} \sqrt{U^{\rm 3B}*\chi}(x_{i}) U^{\rm 3B}(x_{i}-x_{j}) \sqrt{U^{\rm 3B}*\chi}(x_{j}).
\end{align}
It is worth noting that the nonnegativity of $U^{\rm 3B} * \chi$ follows from the assumptions that each term is nonnegative. Moreover, the convolution $U^{\rm 3B} * \chi$ is well defined since $U^{\rm 3B} \in L^{\infty}(\mathbb R^{2})$ (which is the consequence of the assumption that $\widehat{U^{\rm 3B}} \in L^{1}(\mathbb R^{2})$) and $\chi \in L^{1}(\mathbb R^{2})$.

Now, since $U^{\rm 3B}$ is positive definite, in the sense that $\widehat{U^{\rm 3B}} \geq 0$, we have
$$
\iint_{\mathbb R^{4}}f(y)\sqrt{U^{\rm 3B}*\chi}(y)U^{\rm 3B}(y-z)f(z)\sqrt{U^{\rm 3B}*\chi}(z){\rm d}y{\rm d}z = \int_{\mathbb R^{2}}\widehat{U^{\rm 3B}}(k)\left|\widehat{f\sqrt{U^{\rm 3B}*\chi}}(k)\right|^{2} \geq 0.
$$
With the choice of the Radon measure $f$ at the beginning of the proof, this is equivalent to
\begin{align}\label{ineq:three-one-3}
\text{RHS } \eqref{ineq:three-one-2} \geq{} & 2\sum_{1\leq i\leq N} \sqrt{U^{\rm 3B}*\chi}(x_{i}) (U^{\rm 3B}*\chi\sqrt{U^{\rm 3B}*\chi})(x_{i})\nonumber \\
& - \iint_{\mathbb R^{4}}\chi(y)\sqrt{U^{\rm 3B}*\chi}(y)U^{\rm 3B}(y-z)\chi(z)\sqrt{U^{\rm 3B}*\chi}(z){\rm d}y{\rm d}z.
\end{align}
Putting all together \eqref{ineq:three-one-1}, \eqref{ineq:three-one-2}, \eqref{ineq:three-one-3} we obtain the desired inequality \eqref{ineq:three-one}.
\end{proof}

We have reached the position to establish the lower bound on the quantum energy. On the one hand, the kinetic energy is estimated using \eqref{ineq:HO}. On the other hand, the attractive two-body interaction is estimated from below by the one-body term, employing the Levy--Leblond technique \cite{LevLeb-69}. Finally, the three-body interaction is estimated using \Cref{lem:three-one}. Combining these estimates, for any many-body wave function $\Psi_{N}$, by substituting $U^{\rm 3B}$ in \eqref{ineq:three-one} with $U^{\rm 3B}_{N^{\beta}}$ and setting $\chi = N\varrho_{\gamma_{\Psi_{N}}^{(1)}}$, we obtain
\begin{align*}
& \frac{1}{N^{3}}\left\langle \Psi_{N} \left| \sum_{1\leq i \ne j \ne k \leq N}U^{\rm 3B}(x_{i}-x_{j})U^{\rm 3B}(x_{i}-x_{k}) \right| \Psi_{N} \right\rangle \\
& \geq \iint_{\mathbb R^{4}}\varrho_{\gamma_{\Psi_{N}}^{(1)}}(y)\sqrt{U^{\rm 3B}*\varrho_{\gamma_{\Psi_{N}}^{(1)}}}(y)U^{\rm 3B}(y-z)\varrho_{\gamma_{\Psi_{N}}^{(1)}}(z)\sqrt{U^{\rm 3B}*\varrho_{\gamma_{\Psi_{N}}^{(1)}}}(z){\rm d}y{\rm d}z - 3N^{4\beta-1}\|U^{\rm 3B}\|_{L^{\infty}}^{2}.
\end{align*}
In the above, we have simply assumed that the error term in \eqref{ineq:three-one} is estimated by $3N^{2}\|U^{\rm 3B}_{N^{\beta}}\|_{L^{\infty}}^{2} \leq 3N^{4\beta+2}\|U^{\rm 3B}\|_{L^{\infty}}^{2}$. It is worth noting that $U^{\rm 3B} \in L^{\infty}(\mathbb R^{2})$ since $\widehat{U^{\rm 3B}} \in L^{1}(\mathbb R^{2})$. Consequently, we arrive at the lower bound
$$
E^{\rm QM}_{a,b,N} \geq E^{\rm mH}_{a,b,N}[1] - CaN^{2\alpha-1} - CbN^{4\beta-1}.
$$
In this expression, $E^{\rm mH}_{a,b,N}[1]$ denotes the \emph{modified} Hartree energy
\begin{equation}\label{energy:h-modified}
E_{a,b,N}^{\rm mH}[1] := \inf\left\{\mathcal{E}_{a,b,N}^{\rm mH}[v] : v \in H^{1}(\mathbb R^{2}), \int_{\mathbb R^{2}}|v|^{2}=1\right\}
\end{equation}
where the \emph{modified} Hartree functional $\mathcal{E}_{a,b,N}^{\rm mH}$ is defined as
\begin{align}\label{functional:h-modified}
\mathcal{E}_{a,b,N}^{\rm mH}[v] := & \int_{\mathbb R^{2}} |\nabla v(x)|^{2}{\rm d}x - \frac{a}{2}\iint_{\mathbb R^{4}} U^{\rm 2B}_{N^{\alpha}}(x-y)|v(x)|^{2}|v(y)|^{2}{\rm d}x{\rm d}y \nonumber \\
& + \frac{b}{6}\iint_{\mathbb R^{4}}|v(x)|^{2}\sqrt{U^{\rm 3B}_{N^{\beta}}*|v|^{2}}(x)U^{\rm 3B}_{N^{\beta}}(x-y)|v(y)|^{2}\sqrt{U^{\rm 3B}_{N^{\beta}}*|v|^{2}}(y){\rm d}x{\rm d}y.
\end{align}
It is noteworthy that, following the proof of \Cref{lem:cv-hartree-potentials}, $E^{\rm mH}_{a,b,N}[1]$ is essentially equivalent to $E^{\rm H}_{a,b,N}[1]$ in the limit $N\to\infty$. In particular, the results presented in \Cref{thm:behaviors-h} remain valid for $E_{a,b,N}^{\rm mH}[1]$ in \eqref{energy:h-modified} as well as its ground states. Consequently, \Cref{thm:qm} follows from \Cref{thm:behaviors-h,thm:qm-h}.

\appendix
\section{A functional inequality}\label{app:inequality}

In this Appendix, we present a functional inequality, which is instrumental in the proof of \Cref{lem:cv-hartree-potentials} regarding the convergence of the interaction potentials.

\begin{lemma}\label{lem:kinetic}
For every $v \in H^{1}(\mathbb R^{2})$ and $0 \leq g \in L^{1}(\mathbb R^{2})$ with $\int_{\mathbb R^{2}}g = 1$, we have
$$
\left|\nabla \sqrt{g*|v|^{2}}\right| \leq \sqrt{g*|\nabla v|^{2}}.
$$
\end{lemma}

\begin{proof}
We have
$$
\left|\nabla \sqrt{g*|v|^{2}}\right| = \frac{|\nabla (g*|v|^{2})|}{2\sqrt{g*|v|^{2}}} = \frac{|g*|v|\nabla |v||}{\sqrt{g*|v|^{2}}} \leq \frac{g*|v||\nabla v|}{\sqrt{g*|v|^{2}}} \leq \sqrt{g*|\nabla v|^{2}}
$$
where we have used the diamagnetic inequality in the next-to-last inequality and the Cauchy--Schwarz inequality in the last inequality.
\end{proof}

\end{document}